\documentclass[a4paper, UKenglish, cleveref, autoref, thm-restate]{lipics-v2021}
\pdfoutput=1 
\hideLIPIcs  
\title{Massively Parallel Correlation Clustering in Bounded Arboricity Graphs}

\author{M\'{e}lanie Cambus}{Aalto University, Finland}{melanie.cambus@aalto.fi}{https://orcid.org/0000-0002-7635-3924}{This work was supported in part by the Academy of Finland, Grant 334238.}

\author{Davin Choo}{National University of Singapore, Singapore}{davin@u.nus.edu}{https://orcid.org/0000-0002-4545-7341}{This research/project is supported by the National Research Foundation, Singapore under its AI Singapore Programme (AISG Award No: AISG-PhD/2021-08-013).}

\author{Havu Miikonen}{Aalto University, Finland}{havu.miikonen@aalto.fi}{https://orcid.org/0000-0001-5690-9887}{}

\author{Jara Uitto}{Aalto University, Finland}{jara.uitto@aalto.fi}{https://orcid.org/0000-0002-5179-5056}{This work was supported in part by the Academy of Finland, Grant 334238.}

\authorrunning{M.\,Cambus, D.\,Choo, H.\,Miikonen and J.\,Uitto}

\Copyright{M\'{e}lanie Cambus, Davin Choo, Havu Miikonen and Jara Uitto}

\Copyright{M\'{e}lanie Cambus, Davin Choo, Havu Miikonen and Jara Uitto}


\usepackage{algorithm}
\usepackage{algpseudocode}
\usepackage{booktabs}
\usepackage{xspace}
\usepackage{comment}
\renewcommand{\arraystretch}{1.2}

\usepackage{thmtools}
\theoremstyle{plain}
\newtheorem{model}[theorem]{Model}
\newtheorem{question}{Question}
\theoremstyle{remark}

\usepackage{tikz}
\usetikzlibrary{calc, graphs, graphs.standard, shapes, arrows, positioning, decorations.pathreplacing, decorations.markings, decorations.pathmorphing, fit, matrix, patterns, shapes.misc}

\usepackage{mathtools}
\DeclarePairedDelimiter\ceil{\lceil}{\rceil}

\newcommand{\eps}{\varepsilon}
\newcommand{\cC}{\mathcal{C}}
\newcommand{\cO}{\mathcal{O}}
\newcommand{\Abs}[1]{\left\lvert#1\right\rvert}
\newcommand{\Paren}[1]{\left(#1\right)}
\newcommand{\Brac}[1]{\left[#1\right]}
\newcommand{\cA}{\mathcal{A}}

\newcommand{\arb}{\lambda}

\newcommand{\PIVOT}{\texttt{PIVOT}\xspace}
\newcommand{\poly}{\textrm{poly}}

\begin{document}

\maketitle

\begin{abstract}
Identifying clusters of similar elements in a set is a common task in data analysis.
With the immense growth of data and physical limitations on single processor speed, it is necessary to find efficient parallel algorithms for clustering tasks.
In this paper, we study the problem of correlation clustering in bounded arboricity graphs with respect to the Massively Parallel Computation (MPC) model.
More specifically, we are given a complete graph where the edges are either positive or negative, indicating whether pairs of vertices are similar or dissimilar.
The task is to partition the vertices into clusters with as few disagreements as possible.
That is, we want to minimize the number of positive inter-cluster edges and negative intra-cluster edges.

Consider an input graph $G$ on $n$ vertices such that the positive edges induce a $\arb$-arboric graph.
Our main result is a 3-approximation (\emph{in expectation}) algorithm to correlation clustering that runs in $\cO\Paren{\log \arb \cdot \poly\Paren{\log \log n}}$ MPC rounds in the \emph{strongly sublinear memory regime}.
This is obtained by combining structural properties of correlation clustering on bounded arboricity graphs with the insights of Fischer and Noever (SODA '18) on randomized greedy MIS and the \PIVOT algorithm of Ailon, Charikar, and Newman (STOC '05).
Combined with known graph matching algorithms, our structural property also implies an exact algorithm and algorithms with \emph{worst case} $(1+\eps)$-approximation guarantees in the special case of forests, where $\arb=1$.
\end{abstract}

\section{Introduction}

Graphs are a versatile abstraction of datasets and clustering on graphs is a common unsupervised machine learning task for data-analytical purposes such as community detection and link prediction~\cite{ChenNIPS,Bianchi2013}.

Here, we study the correlation clustering problem which aims at grouping elements of a dataset according to their similarities.
Consider the setting where we are given a complete signed graph $G = (V, E = E^+ \cup E^-)$ where edges are given positive ($E^+$) or negative ($E^-$) labels, signifying whether two points are similar or not.
The task is to find a partitioning of the vertex set $V$ into clusters $C_1, C_2, \ldots, C_r$, where $r$ is not fixed by the problem statement but can be chosen freely by the algorithm.\footnote{This is in contrast to, for example, the classic $k$-means clustering where $k$ is an input problem parameter.}
If endpoints of a positive edge belong to the same cluster, we say that the edge is a \emph{positive agreement}; and a \emph{positive disagreement} otherwise.
Meanwhile, if endpoints of a negative edge belong to the same cluster, we say that the edge is a \emph{negative disagreement}; and a \emph{negative agreement} otherwise.
The goal of correlation clustering is to obtain a clustering that maximizes agreements or minimizes disagreements.

As pointed out by Chierichetti, Dalvi and Kumar~\cite{chierichetti2014mapreduce}, the positive degrees of vertices are typically bounded in many applications.
This motivates the study of parallel algorithms for correlation clustering as a function of the maximum degree of the input graph.
However, many real life networks, such as those modelled by scale-free network models (such as Barab\'{a}si-Albert), admit structures with a few high degree nodes and a small average degree.
To capture such graphs, we generalize the study of bounded degree graphs to the study of low arboricity graphs in this work.
In particular, we focus on the case of minimizing disagreements when the positive edges of the input graph induces a $\arb$-arboric graph.

In the complete signed graph setting, one can perform cost-charging arguments via ``bad triangles'' to prove approximation guarantees.
A set of 3 vertices $\{u, v, w\}$ is a \emph{bad triangle} if $\{u,v\}, \{v,w\} \in E^+$ and $\{u,w\} \in E^-$.
As edges of any bad triangle induce at least one disagreement in any clustering, one can lower bound the cost of any optimum clustering by the number of \emph{edge-disjoint} bad triangles in the input graph.
\PIVOT~\cite{ailon2008aggregating} is a well-known algorithm that provides a 3-approximation (in expectation) to the problem of minimizing disagreements in the sequential setting by using a cost-charging argument on bad triangles.
It works as follows: as long as the graph is non-empty, pick a vertex $v$ uniformly at random and form a new cluster using $v$ and its ``positive neighbors'' (i.e.\ joined by a positive edge).
One can view \PIVOT as simulating greedy MIS with respect to a uniform-at-random permutation of vertices.\footnote{A subset $M \subseteq V$ is a maximal independent set (MIS) if (1) for any two vertices $u, v \in M$, $u$ and $v$ are not neighbors, and (2) for any vertex $v \in V$, either $v \in M$ or $v$ has a neighbor in $M$. Given a vertex ordering $\pi : [n] \rightarrow V$, greedy MIS refers to the process of iterating through $\pi(1), \ldots, \pi(n)$ and adding each vertex to $M$ if it has no neighbor of smaller ordering.}

Many of the known distributed algorithms for the correlation clustering problem adapt the \PIVOT algorithm.
The basic building block is to fix a random permutation and to create the clusters by finding, in parallel, local minimums according to the permutation.
The \texttt{ParallelPIVOT}, \texttt{C4} and \texttt{ClusterWild!} algorithms~\cite{chierichetti2014mapreduce,pan2015parallel} all obtain constant approximations in $\cO(\log n \cdot \log \Delta)$ synchronous rounds, where $\Delta$ stands for the maximum positive degree.\footnote{Technically speaking, \texttt{ParallelPIVOT} does not compute a greedy MIS. Instead, it computes random independent sets in each phase and only uses the initial random ordering to perform tie-breaking. i.e.\ if a vertex $u$ has more than one positive neighbor in the independent set, then vertex $u$ joins the cluster defined by the neighbor with the smallest assigned order.}
Meanwhile, with a tighter analysis of randomized greedy MIS algorithm~\cite{fischer2018tight}, one can obtain a $3$-approximation in $\cO(\log n)$ rounds by directly simulating \texttt{PIVOT}.
All above approximation guarantees are in expectation.

\subsection{Computational model}
We consider the \emph{Massive Parallel Computation (MPC) model}~\cite{karloff2010model,beame2017communication} which serves as a theoretical abstraction of several popular massively parallel computation frameworks such as Dryad~\cite{isard2007dryad}, Hadoop~\cite{white2012hadoop}, MapReduce~\cite{dean2008mapreduce}, and Spark~\cite{zaharia2010spark}.

In the MPC model, we have $M$ machines, each with memory size $S$, and we wish to solve a problem given an input of size $N$.
In the context of correlation clustering, we may think of $N = |E^+|$, since the negative edges can be inferred from missing positive edges.
Typically, the \emph{local} memory bound $S$ is assumed to be significantly smaller than $N$.
We focus on the \emph{strongly sublinear memory} regime, where $S = \widetilde{\cO}\Paren{n^\delta}$ for some constant $\delta < 1$.
Ideally, the total memory $S \cdot M$ is not much larger than $N$.

The computation in the MPC model proceeds in \emph{synchronous rounds}.
In each round, each machine can perform arbitrary computation on the data that resides on it.\footnote{Although there is no hard computation constraint in the MPC model, all known MPC algorithms spend polynomial time on each machine in any given round.}
Then, each machine communicates in an all-to-all fashion with all other machines conditioned on sending and receiving messages of size at most $\cO(S)$.
This concludes the description of an MPC round.
Since communication costs are typically the bottleneck, the metric for evaluating the efficiency of an MPC algorithm is the \emph{number of rounds} required.

\subsection{Our contributions}
\label{sec:results}

Our goal is to obtain efficient algorithms  for correlation clustering in the sublinear memory regime of MPC (see \cref{model:sublinear}) when given a complete signed graph $G$, with maximum positive degree $\Delta$, where the set of positive edges $E^+$ induces a $\arb$-arboric graph.
Our main contributions are the following:
\begin{enumerate}
    \item By combining known techniques, we show that one can compute a randomized greedy MIS, with respect to a uniform-at-random permutation of vertices, in $\cO\Paren{\log \Delta \cdot \log^3 \log n}$ MPC rounds.
    If we allow extra global memory (see \cref{model:sublinear-extra}), this can be sped up to $\cO\Paren{\log \Delta \cdot \log \log n}$ MPC rounds.
    See \cref{thm:randomized-greedy-MIS-informal} for details.\\
    We believe that this result is of independent interest beyond applications to correlation clustering.
    To the best of our knowledge, our algorithm for greedy MIS improves upon the state-of-the-art for any $\Delta \in o(n^{1 / \log^3 \log n})$.
    
    \item Our main result (\cref{thm:algo-implication}) is that one can effectively ignore vertices of degrees larger than $\cO(\arb)$ when computing a correlation clustering.
    Then, the overall runtime and approximation guarantees are inherited from the choice of algorithm used to solve correlation clustering on the remaining bounded degree subgraph.\footnote{In some works, ``bounded degree'' is synonymous with ``maximum degree $\cO(1)$''. Here, we mean that the maximum degree is $\cO(\arb)$.}
    \item Using our main result, we show how to obtain efficient correlation clustering algorithms for bounded arboricity graphs.
    By simulating \PIVOT on a graph with maximum degree $\cO(\arb)$ via \cref{thm:randomized-greedy-MIS-informal}, we get
    \begin{enumerate}[(i)]
        \item A 3-approximation (in expectation) algorithm in $\cO\Paren{\log \arb \cdot \log^3 \log n}$ MPC rounds.
        \item A 3-approximation (in expectation) algorithm in $\cO\Paren{\log \arb \cdot \log \log n}$ MPC rounds, possibly using extra global memory.
    \end{enumerate}
    In the special case of forests (where $\arb = 1$), we show that the optimum correlation clustering is equivalent to computing a \emph{maximum matching}.
    Let $0 < \eps \leq 1$ be a constant.
    By invoking three different known algorithms (one for \emph{maximum matching} and two for \emph{maximal matching}), and hiding $1/\eps$ factors in $\cO_{\eps}(\cdot)$, we obtain
    \begin{enumerate}[(i)]\setcounter{enumii}{2}
        \item An exact randomized algorithm that runs in $\widetilde{\cO}\Paren{\log n}$ MPC rounds.
        \item A $(1+\eps)$-approx.\ (worst case) det.\ algo.\ that runs in $\cO_{\eps}\Paren{\log \log^* n}$ MPC rounds.
        \item A $(1+\eps)$-approx.\ (worst case) randomized algo.\ that runs in $\cO_{\eps}(1)$ MPC rounds.
    \end{enumerate}
    Finally, for low-arboricity graphs, the following result may be of interest:
    \begin{enumerate}[(i)]\setcounter{enumii}{5}
        \item An $\cO\Paren{\arb^2}$-approx.\ (worst case) deterministic algo.\ that runs in $\cO(1)$ MPC rounds.
    \end{enumerate}
\end{enumerate}
For more details and an in-depth discussion about our techniques, see \cref{sec:techniques}.

\subsection{Outline and notation}

\subsubsection{Outline}

Before diving into formal details, we highlight the key ideas behind our results in \cref{sec:techniques}.
\cref{sec:randomized-greedy-mis} shows how to efficiently compute a randomized greedy MIS.
Our structural result about correlation clustering in bounded arboricity graphs is presented in \cref{sec:structural}.
We combine this structural insight with known algorithms in \cref{sec:application} to yield efficient correlation clustering algorithms.
Finally, we conclude with some open questions in \cref{sec:conclusion}.

Due to space constraints, we will defer some proof details to \cref{sec:deferred-proofs}.

\subsubsection{Notation}

In this work, we only deal with complete signed graphs $G = (V, E = E^+ \cup E^-)$ where $\Abs{V} = n$, $\Abs{E} = \binom{n}{2}$, and $E^+$ and $E^-$ denote the sets of positively and negatively labeled edges respectively.
For a vertex $v$, the sets $N^+(v) \subseteq V$ and $N^-(v) \subseteq V$ denote vertices that are connected to $v$ via positive and negative edges, respectively.
We write $\Delta = \max_{v \in V} \Abs{N^+(v)}$ as the maximum \emph{positive} degree in the graph.
The $k$-hop neighborhood of a vertex $v$ is the set of vertices that have a path from $v$ involving at most $k$ \emph{positive} edges.

A clustering $\cC$ is a partition of the vertex set $V$.
That is, $\cC$ is a set of sets of vertices such that (i) $A \cap B = \emptyset$ for any two sets $A,B \in \cC$ and (ii) $\cup_{A \in \cC} A = V$.
For a cluster $C \subseteq V$, $N^+_C(v) = N^+(v) \cap C$ is the set of positive neighbors of $v$ that lie within cluster $C$.
We write $d^+_C(v) = \Abs{N^+(v) \cap C}$ to denote the positive degree of $v$ within $C$.
If endpoints of a positive edge \emph{do not} belong to the same cluster, we say that the edge is a \emph{positive disagreement}.
Meanwhile, if endpoints of a negative edge belong to the same cluster, we say that the edge is a \emph{negative disagreement}.
Given a clustering $\cC$, the cost of a clustering $cost(\cC)$ is defined as the total number of disagreements.

The arboricity $\arb_G$ of a graph $G = (V,E)$ is defined as $\arb_G = \max_{S \subseteq V} \left\lceil \frac{\Abs{E(S)}}{\Abs{S}-1} \right\rceil$, where $E(S)$ is the set of edges induced by $S \subseteq V$.
We drop the subscript $G$ when it is clear from context.
A graph with arboricity $\arb$ is said to be $\arb$-arboric.
We denote the set $\{1, 2, \ldots, n\}$ by $[n]$.
We hide absolute constant multiplicative factors and multiplicative factors logarithmic in $n$ using standard notations: $\cO(\cdot)$, $\Omega(\cdot)$, and $\widetilde{\cO}(\cdot)$.
The notation $\log^* n$ refers to the smallest integer $t$ such that the $t$-iterated logarithm of $n$ is at most 1.\footnote{That is, $\log^{(t)} n \leq 1$. For all practical values of $n$, one may treat $\log^* n \leq 5$.}
An event $\mathcal{E}$ on a $n$-vertex graph holds with high probability if it happens with probability at least $1 - n^{-c}$ for an arbitrary constant $c > 1$, where $c$ may affect other constants (e.g.\ those hidden in the asymptotics).

We now fix the parameters in our model of computation.
\cref{model:sublinear} is the standard definition of strongly sublinear MPC regime while \cref{model:sublinear-extra} is a relaxed variant which guarantees that there are at least $M \geq n$ machines.
While the latter model may utilize more global memory than the standard strongly sublinear regime, it facilitates conceptually simpler algorithms.

\begin{model}[Strongly sublinear MPC regime]
\label{model:sublinear}
Consider the MPC model.
The input graph with $n$ vertices is of size $N \in \Omega(n)$.
We have $M \in \Omega\Paren{N/S}$ machines, each having memory size $S \in \widetilde{\cO}\Paren{n^{\delta}}$, for some constant $0 < \delta < 1$.
The total global memory usage is $M \cdot S \geq N$.
\end{model}

\begin{model}[Strongly sublinear MPC regime with at least $n$ machines]
\label{model:sublinear-extra}
Consider the MPC model.
The input graph with $n$ vertices is of size $N \in \Omega(n)$.
We have $M \geq n$ machines and each vertex is given access to a machine with memory size $S \in \widetilde{\cO}\Paren{n^{\delta}}$, for some constant $0 < \delta < 1$.
The total global memory usage is $\max\{\widetilde{\cO}\Paren{n^{1+\delta}}, M \cdot S\} \geq N$.
\end{model}

To avoid unnecessary technical complications for \cref{model:sublinear-extra}, we assume throughout the paper that $\Delta \in \cO(S)$.
This assumption can be lifted using the virtual communication tree technique described by Ghaffari and Uitto~\cite{ghaffari2019sparsifying}.

\begin{remark}[Role and motivation for \cref{model:sublinear-extra}]
From an algorithmic design perspective, the slightly relaxed \cref{model:sublinear-extra} allows one to focus on keeping the amount of ``local memory required by each vertex'' to the sublinear memory regime.
Oftentimes\footnote{
As a warm-up description of their algorithm, Ghaffari and Uitto~\cite[Assumption (2) on page 6]{ghaffari2019sparsifying} uses more machines than just $M = N/S$.
Meanwhile, using slightly more global memory, the algorithm of ASSWZ~\cite{andoni2018parallel} is straightforward to understand (e.g. see Ghaffari~\cite[Section 3.3]{MPAnotes}) and can achieve conjecturally optimal running time, with respect to the $\Omega(\log D)$ conditional lower bound for solving graph connectivity in MPC via the 2-cycle problem.
}, algorithms are first described in relaxed models (such as \cref{model:sublinear-extra}, or by simply allowing more total global memory used) with a simple-to-understand analysis before using further complicated argument/analysis to show that it in fact also works in \cref{model:sublinear}.\footnote{
In our case, we first designed \cref{alg:prefix-faster} in \cref{model:sublinear-extra} but were unable to show that it also works in \cref{model:sublinear}.
Thus, we designed \cref{alg:prefix-slower} that works in \cref{model:sublinear}.
However, we decided to keep the description and analysis of the simpler \cref{alg:prefix-faster} in the paper -- it is algorithmically very clean and we hope that it is easier to understand the technicalities of the more involved \cref{alg:prefix-faster} after seeing the structure and analysis of the simpler \cref{alg:prefix-slower}.
}
\end{remark}

\subsection{Further related work}
\label{sec:related}

Correlation clustering on complete signed graphs was introduced by Bansal, Blum and Chawla~\cite{bansal2004correlation}.\footnote{For relevant prior work, we try our best to list all authors when there are three or less, and use their initials when there are more (e.g.\ CMSY, PPORRJ, BBDFHKU). While this avoids the use of et al.\ in citations in favor of an equal mention of all authors' surnames, we apologize for the slight unreadability.}
They showed that computing the optimal solution to correlation clustering is NP-complete, and explored two different optimization problems: maximizing agreements, or minimizing disagreements.
While the optimum clusterings to both problems are the same (i.e.\ a clustering minimizes disagreements if and only if it maximizes agreements), the complexity landscapes of their approximate versions are wildly different.

Maximizing agreements is known to admit a polynomial time approximation scheme in complete graphs~\cite{bansal2004correlation}.
Furthermore, Swamy~\cite{swamy2004semidefinite} gave a 0.7666-approximation on general weighted graphs via semidefinite programming.

On the other hand, for minimizing disagreements, the best known approximation ratio for complete graphs is 2.06, due to CMSY~\cite{chawla2015near}, via probabilistic rounding of a linear program (LP) solution.
This 2.06-approximation uses the same LP as the one proposed by Ailon, Charikar and Newman~\cite{ailon2008aggregating} but performs probabilistic rounding more carefully, nearly matching the integrality gap of 2 shown by Charikar, Guruswami and Wirth~\cite{charikar2005clustering}.
In general weighted graphs, the current state of the art, due to DEFI~\cite{demaine2006general_weighted}, gives an $\cO(\log n)$-approximation through an LP rounding scheme.

In a distributed setting, PPORRJ~\cite{pan2015parallel} presented two random algorithms (\texttt{C4} and \texttt{ClusterWild!}) to address the correlation clustering problem in the case of complete graphs, aiming at better time complexities than \texttt{KwikCluster}.
The \texttt{C4} algorithm gives a 3-approximation in expectation, with a polylogarithmic number of rounds where the greedy MIS problem is solved on each round.
The \texttt{ClusterWild!} algorithm gives up on the independence property in order to speed up the process, resulting in a $(3 + \eps)$-approximation.
Both those algorithms are proven to terminate after $\cO\Paren{\frac{1}{\epsilon} \cdot \log n \cdot \log \Delta}$ rounds with high probability. 
A third distributed algorithm for solving correlation clustering is given by Chierichetti, Dalvi and Kumar~\cite{chierichetti2014mapreduce} for the MapReduce model.
Their algorithm, \texttt{ParallelPivot}, also gives a constant approximation in polylogarithmic time, without solving a greedy MIS in each round.
Using a tighter analysis, Fischer and Noever~\cite{fischer2018tight} showed that randomized greedy MIS terminates in $\cO(\log n)$ round with high probability, which directly implies an $\cO(\log n)$ round simulation of \PIVOT in various distributed computation models.

For our approach, the \emph{randomized greedy MIS} plays a crucial role in terms of the approximation ratio.
Blelloch, Fineman and Shun~\cite{Blelloch2012} showed that randomized greedy MIS terminates in $\cO(\log^2 n)$ parallel rounds with high probability.
This was later improved to $\cO(\log n)$ rounds by Fischer and Noever~\cite{fischer2018tight}.
Faster algorithms are known for finding an MIS that may not satisfy the greedy property.
For example, Ghaffari and Uitto~\cite{ghaffari2019sparsifying} showed that there is an MIS algorithm running in $\cO\Paren{\sqrt{\log \Delta} \cdot \log \log \Delta + \sqrt{\log \log n}}$ MPC rounds.
This algorithm was later adapted to bounded arboricity with runtime of $\cO\Paren{\sqrt{\log \lambda} \cdot \log \log \lambda + \log^2 \log n}$ by BBDFHKU~\cite{Behnezhad2019} and improved to $\cO\Paren{\sqrt{\log \lambda} \cdot \log \log \lambda + \log \log n}$ by Ghaffari, Grunau and Jin~\cite{GGJ20}.
There is also a deterministic MIS algorithm that runs in $\cO\Paren{\log \Delta + \log \log n}$ MPC rounds due to Czumaj, Davies and Parter~\cite{Czumaj2020}.

\section{Techniques}
\label{sec:techniques}

In this section, we highlight the key ideas needed to obtain our results described in \cref{sec:results}.
We begin by explaining some computational features of the MPC model so as to set up the context needed to appreciate our algorithmic results.
By exploiting these computational features together with a structural result of randomized greedy MIS by Fischer and Noever~\cite{fischer2018tight}, we explain how to compute a randomized greedy MIS in $\cO\Paren{\log \Delta \cdot \log \log n}$ MPC rounds.
We conclude this section by explaining how to obtain our correlation clustering results by using our structural lemma that reduces the maximum degree of the input graph to $\cO(\arb)$.

\subsection{Computational features of MPC}
\label{sec:MPC-features}

\subsubsection{Detour: The classical models of LOCAL and CONGEST}

To better appreciate of the computational features of MPC, we first describe the classical distributed computational models of LOCAL and CONGEST~\cite{linial1992locality, peleg2000distributed}.

In the LOCAL model, all vertices are treated as individual computation nodes and are given a unique identifier -- some binary string of length $\cO (\log n)$.
Computation occurs in synchronous rounds where each vertex does the following: perform arbitrary local computations, then send messages (of unbounded size) to neighbors.
As the LOCAL model does not impose any restrictions on computation or communication costs (beyond a topological restriction), the performance of LOCAL algorithms is measured in the number of rounds used.
Furthermore, since nodes can send unbounded messages, every vertex can learn about its $k$-hop neighborhood in $k$ LOCAL rounds.

The CONGEST model is identical to the LOCAL model with an additional restriction: the size of messages that can be sent or received per round can only be $\cO (\log n)$ bits across each edge.
This means that CONGEST algorithms may no longer assume that they can learn about the $k$-hop topology for every vertex in $k$ CONGEST rounds.

Since the MPC model does not restrict computation within a machine, one can directly simulate any $k$-round LOCAL or CONGEST algorithm in $\cO(k)$ MPC rounds, as long as each machine sends and receives messages of size at most $\cO(S)$.
This often allows us to directly invoke existing LOCAL and CONGEST algorithms in a black-box fashion.

\subsubsection{Round compression}

First introduced by C{\L}MMOSP~\cite{czumaj2019round}, the goal of round compression is to simulate multiple rounds of an iterative algorithm within a single MPC round.
To do so, one gathers ``sufficient amount of information'' into a single machine.
For example, if an iterative algorithm $\cA$ only needs to know the $k$-hop neighborhood to perform $r$ steps of an algorithm, then these $r$ steps can be compressed into a single MPC round once the $k$-hop neighborhood has been gathered.

\subsubsection{Graph exponentiation}

One way to speed up computation in an all-to-all communication setting (such as MPC) is the well-known graph exponentiation technique of Lenzen and Wattenhofer~\cite{lenzen2010brief}.
The idea is as follows: Suppose each vertex is currently aware of its $2^{k-1}$-hop neighborhood, then by sending this $2^{k-1}$ topology to all their current neighbors, each vertex learns about their respective $2^k$-hop neighborhoods in one additional MPC round.
In other words, every vertex can learn about its $k$-hop neighborhood in $\cO(\log k)$ MPC rounds, as long as the machine memory is large enough.
See \cref{fig:graph-exp} for an illustration.
This technique is motivated by the fact that once a vertex has gathered its $k$-hop neighborhood, it can execute any LOCAL algorithm that runs in $k$ rounds in just a single MPC round.

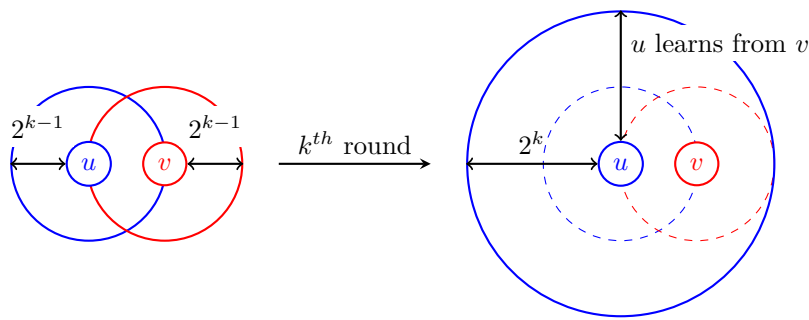
\begin{figure}[htbp]
\centering
\begin{tikzpicture}
\node[draw, circle, blue, thick, minimum size=58pt] at (0,0) (u-before-circle) {};
\node[draw, circle, red, thick, minimum size=58pt] at (1,0) (v-before-circle) {};
\node[draw, circle, blue, thick, minimum size=5pt, fill=white] at (0,0) (u-before) {$u$};
\node[draw, circle, red, thick, minimum size=5pt, fill=white] at (1,0) (v-before) {$v$};

\node[draw, circle, blue, dashed, minimum size=58pt] at (7,0) (u-after-circle) {};
\node[draw, circle, red, dashed, minimum size=58pt] at (8,0) (v-after-circle) {};
\node[draw, circle, blue, thick, minimum size=115pt] at (7,0) (uv-after-circle) {};
\node[draw, circle, blue, thick, minimum size=5pt, fill=white] at (7,0) (u-after) {$u$};
\node[draw, circle, red, thick, minimum size=5pt, fill=white] at (8,0) (v-after) {$v$};

\draw[thick, -stealth] (2.5,0) -- node[above, midway]{$k^{th}$ round} (4.5,0);
\draw[thick, <->] (u-before.west) -- node[above=6pt, midway, fill=white]{$2^{k-1}$} (u-before-circle.west);
\draw[thick, <->] (v-before.east) -- node[above=6pt, midway, fill=white]{$2^{k-1}$} (v-before-circle.east);
\draw[thick, <->] (u-after.west) -- node[above, midway]{$2^k$} (uv-after-circle.west);
\draw[thick, <->] (u-after.north) -- node[right, pos=0.75, fill=white]{$u$ learns from $v$} (uv-after-circle.north);
\end{tikzpicture}
\caption{
After round $k$, vertex $u$ knows the graph topology within its $2^{k}$-hop neighborhood.
}
\label{fig:graph-exp}
\end{figure}

\subsubsection{Combining graph exponentiation with round compression}
\label{sec:combination}

Suppose we wish to execute a $k$-round LOCAL algorithm but the machine memory of a single machine is too small to contain entire $k$-hop neighborhoods.
To get around this, one can combine graph exponentiation with round compression:
\begin{enumerate}
    \item All vertices collect the largest possible neighborhood using graph exponentiation.
    \item Suppose $\ell$-hop neighborhoods were collected, for some $\ell < k$.
    All vertices simulate $\ell$ steps of the LOCAL algorithm in a single MPC round using round compression.
    \item All vertices update their neighbors about the status of their computation.
    \item Repeat steps 2-3 for $\cO(k/\ell)$ phases.
\end{enumerate}
This essentially creates a \emph{virtual communication graph} where vertices are connected to their $\ell$-hop neighborhoods.
This allows a vertex to derive, in one round of MPC, all the messages that reaches it in the next $\ell$ rounds of message passing.
Using one more MPC round and the fact that local computation is unbounded, a vertex can inform all its neighbors in the virtual graph about its current state in the simulated message passing algorithm.
See \cref{fig:combined}. 

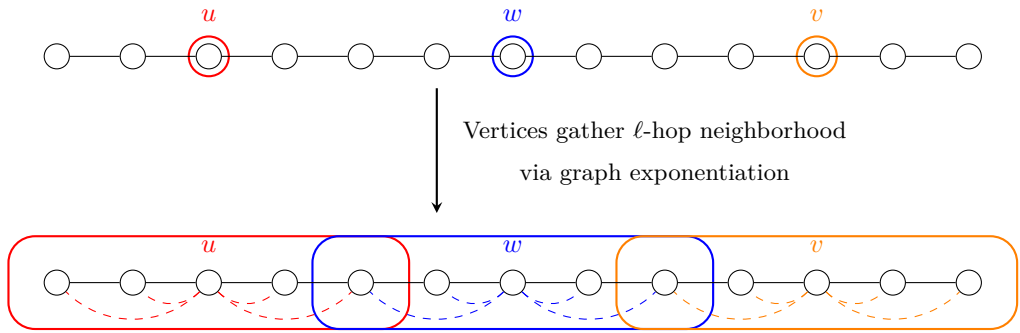
\begin{figure}[htbp]
\centering
\begin{tikzpicture}
\node[draw, circle, minimum size=5pt] at (0,0) (v1-top) {};
\node[draw, circle, minimum size=5pt] at (1,0) (v2-top) {};
\node[draw, circle, minimum size=5pt] at (2,0) (v3-top) {};
\node[draw, circle, minimum size=5pt] at (3,0) (v4-top) {};
\node[draw, circle, minimum size=5pt] at (4,0) (v5-top) {};
\node[draw, circle, minimum size=5pt] at (5,0) (v6-top) {};
\node[draw, circle, minimum size=5pt] at (6,0) (v7-top) {};
\node[draw, circle, minimum size=5pt] at (7,0) (v8-top) {};
\node[draw, circle, minimum size=5pt] at (8,0) (v9-top) {};
\node[draw, circle, minimum size=5pt] at (9,0) (v10-top) {};
\node[draw, circle, minimum size=5pt] at (10,0) (v11-top) {};
\node[draw, circle, minimum size=5pt] at (11,0) (v12-top) {};
\node[draw, circle, minimum size=5pt] at (12,0) (v13-top) {};
\node[draw, thick, circle, minimum size=15pt, red] at (2,0) {};
\node[draw, thick, circle, minimum size=15pt, blue] at (6,0) {};
\node[draw, thick, circle, minimum size=15pt, orange] at (10,0) {};

\node[draw, circle, minimum size=5pt] at (0,-3) (v1-btm) {};
\node[draw, circle, minimum size=5pt] at (1,-3) (v2-btm) {};
\node[draw, circle, minimum size=5pt] at (2,-3) (v3-btm) {};
\node[draw, circle, minimum size=5pt] at (3,-3) (v4-btm) {};
\node[draw, circle, minimum size=5pt] at (4,-3) (v5-btm) {};
\node[draw, circle, minimum size=5pt] at (5,-3) (v6-btm) {};
\node[draw, circle, minimum size=5pt] at (6,-3) (v7-btm) {};
\node[draw, circle, minimum size=5pt] at (7,-3) (v8-btm) {};
\node[draw, circle, minimum size=5pt] at (8,-3) (v9-btm) {};
\node[draw, circle, minimum size=5pt] at (9,-3) (v10-btm) {};
\node[draw, circle, minimum size=5pt] at (10,-3) (v11-btm) {};
\node[draw, circle, minimum size=5pt] at (11,-3)(v12-btm) {};
\node[draw, circle, minimum size=5pt] at (12,-3) (v13-btm) {};
\node[draw, thick, rectangle, rounded corners=10pt, minimum width=150pt, minimum height=35pt, red] at (2,-3) {};
\node[draw, thick, rectangle, rounded corners=10pt, minimum width=150pt, minimum height=35pt, blue] at (6,-3) {};
\node[draw, thick, rectangle, rounded corners=10pt, minimum width=150pt, minimum height=35pt, orange] at (10,-3) {};

\draw[] (v1-top) -- (v2-top);
\draw[] (v2-top) -- (v3-top);
\draw[] (v3-top) -- (v4-top);
\draw[] (v4-top) -- (v5-top);
\draw[] (v5-top) -- (v6-top);
\draw[] (v6-top) -- (v7-top);
\draw[] (v7-top) -- (v8-top);
\draw[] (v8-top) -- (v9-top);
\draw[] (v9-top) -- (v10-top);
\draw[] (v10-top) -- (v11-top);
\draw[] (v11-top) -- (v12-top);
\draw[] (v12-top) -- (v13-top);

\draw[dashed, red] (v3-btm) to[out=-135,in=-45] (v1-btm);
\draw[dashed, red] (v3-btm) to[out=-135,in=-45] (v2-btm);
\draw[dashed, red] (v3-btm) to[out=-45,in=-135] (v4-btm);
\draw[dashed, red] (v3-btm) to[out=-45,in=-135] (v5-btm);
\draw[dashed, blue] (v7-btm) to[out=-135,in=-45] (v5-btm);
\draw[dashed, blue] (v7-btm) to[out=-135,in=-45] (v6-btm);
\draw[dashed, blue] (v7-btm) to[out=-45,in=-135] (v8-btm);
\draw[dashed, blue] (v7-btm) to[out=-45,in=-135] (v9-btm);
\draw[dashed, orange] (v11-btm) to[out=-135,in=-45] (v9-btm);
\draw[dashed, orange] (v11-btm) to[out=-135,in=-45] (v10-btm);
\draw[dashed, orange] (v11-btm) to[out=-45,in=-135] (v12-btm);
\draw[dashed, orange] (v11-btm) to[out=-45,in=-135] (v13-btm);

\draw[] (v1-btm) -- (v2-btm);
\draw[] (v2-btm) -- (v3-btm);
\draw[] (v3-btm) -- (v4-btm);
\draw[] (v4-btm) -- (v5-btm);
\draw[] (v5-btm) -- (v6-btm);
\draw[] (v6-btm) -- (v7-btm);
\draw[] (v7-btm) -- (v8-btm);
\draw[] (v8-btm) -- (v9-btm);
\draw[] (v9-btm) -- (v10-btm);
\draw[] (v10-btm) -- (v11-btm);
\draw[] (v11-btm) -- (v12-btm);
\draw[] (v12-btm) -- (v13-btm);
\draw[thick, -stealth] ($(v6-top.south) + (0,-0.25)$) -- node[right, midway]{\begin{tabular}{c}Vertices gather $\ell$-hop neighborhood\\ via graph exponentiation\end{tabular}} ($(v6-btm.north) + (0,0.75)$);

\node[red, above=5pt of v3-top] {$u$};
\node[blue, above=5pt of v7-top] {$w$};
\node[orange, above=5pt of v11-top] {$v$};
\node[red, above=3pt of v3-btm] {$u$};
\node[blue, above=3pt of v7-btm] {$w$};
\node[orange, above=3pt of v11-btm] {$v$};
\end{tikzpicture}
\caption{
Suppose $\ell=2$.
After each vertex collects their $\ell$-hop neighborhood, computation within each collected neighborhood can be performed in a single compressed MPC round.
While the vertices $u$ and $v$ were originally 8 hops apart, they can communicate in 2 MPC rounds through vertex $w$'s collected neighborhood in the \emph{virtual communication graph}.
Observe that this virtual communication graph has a smaller effective diameter compared to the original input graph.
}
\label{fig:combined}
\end{figure}

\begin{remark}
In \cref{sec:combination}, we make the implicit assumption that the states of the vertices are small and hence can be communicated with small messages.
In many algorithms (e.g.\ for solving MIS, matching, coloring), including ours, the vertices maintain very small states.
Hence, we omit discussion of individual message sizes in the scope of this paper.
\end{remark}

\subsubsection{Broadcast / Convergecast trees}
\label{sec:broadcast-convergecast}

Broadcast trees are a useful MPC data structure introduced by Goodrich, Sitchinava and Zhang~\cite{goodrich2011sorting} that allow us to perform certain aggregation tasks in $\cO\Paren{1/\delta}$ MPC rounds, which is essentially $\cO(1)$ for constant $\delta$.
Suppose we have $\cO(N)$ global memory and $S = \cO\Paren{n^{\delta}}$ local memory.\footnote{We borrow some notation from Ghaffari and Nowicki \cite[Lemma 3.5]{ghaffari2020massively}. For $n$-vertex graphs, $N \in \cO(n^2)$.}
We build an $S$-ary virtual communication tree over the machines.
That is, within one MPC round, the parent machine can send $\cO(1)$ numbers to each of its $S$ children machines, or collect one number from each of its $S$ children machines.
In $\cO\Paren{\log_S N} \subseteq \cO\Paren{1/\delta}$ rounds, for all vertices $v$ in parallel, one can:
\begin{itemize}
    \item broadcast a message from v to all neighboring vertices in $N(v)$;
    \item compute $f(N(v))$, the value of a distributive aggregate function $f$ on set of vertices $N(v)$.
\end{itemize}
An example of such a function $f$ is computing the sum/min/max of numbers that were originally distributed across all machines.
We use broadcast trees in the MPC implementation of the algorithm described in \cref{cor:constant}.

\subsection{Randomized greedy MIS on bounded degree graphs}

The following result of Fischer and Noever~\cite{fischer2018tight} states that each vertex only needs the ordering of the vertices within its $\cO(\log n)$-hop neighborhood in order to compute its own output status within a randomized greedy MIS run.\footnote{More specifically, they analyzed the ``longest length of a dependency path'' and showed that it is $\cO(\log n)$ with high probability, which implies \cref{thm:dependency}.}

\begin{theorem}[Fischer and Noever~\cite{fischer2018tight}]
\label{thm:dependency}
Given a uniform-at-random ordering of vertices, with high probability, the MIS status of any vertex is determined by the vertex orderings within its $\cO(\log n)$-hop neighborhood.
\end{theorem}

Let $\pi: [n] \rightarrow V$ be a uniform-at-random ordering of vertices and $G$ be a graph with maximum degree $\Delta$.
In \cref{sec:randomized-greedy-mis}, we show that one can compute greedy MIS (with respect to $\pi$) in $\cO\Paren{\log \Delta \cdot \log \log n}$ MPC rounds.

\begin{theorem}[Randomized greedy MIS (Informal)]
\label{thm:randomized-greedy-MIS-informal}
Let $G$ be a graph with maximum degree $\Delta$.
Then, randomized greedy MIS can be computed in $\cO\Paren{\log \Delta \cdot \log^3 \log n}$ MPC rounds in \cref{model:sublinear}, or in $\cO\Paren{\log \Delta \cdot \log \log n}$ MPC rounds in \cref{model:sublinear-extra}.
\end{theorem}

\cref{alg:greedy} works in phases.
In each phase, we process a prefix graph $G_{\text{prefix}}$ defined by vertices indexed by a prefix of $\pi$, where the maximum degree is $\cO(\log n)$ by Chernoff bounds.
\cref{alg:prefix-slower} and \cref{alg:prefix-faster} are two subroutines to process prefix graph $G_{\text{prefix}}$.
The latter subroutine is faster by a $\log^2 \log n$ factor but assumes access to more machines.
For a sufficiently large prefix of $\pi$, the maximum degree of the input graph after processing $G_{\text{prefix}}$ drops to $\Delta/2$ with high probability.
This concludes a phase.
Since the maximum degree in the original graph is halved, we can process more vertices in subsequent phases and thus process all $n$ vertices after $\cO(\log \Delta)$ phases.
See \cref{fig:chunking} for an illustration.

\begin{remark}[Discussion about maximum degree]
We implicitly assume that $\Delta > 1$, which can be checked in $\cO(1)$ rounds.
Otherwise, when $\Delta = 1$, the graph only contain pairs of vertices and isolated vertices and greedy MIS can be trivially simulated in one round.
\end{remark}

\begin{algorithm}[htbp]
\caption{Greedy MIS in sublinear memory regime of the MPC model}
\label{alg:greedy}
\begin{algorithmic}[1]
    \State \textbf{Input}: Graph $G = (V, E)$ with maximum degree $\Delta$
    \State Let $\pi: [n] \rightarrow V$ be an ordering of vertices chosen uniformly at random.
    \For{$i = 0, 1, 2, \ldots, \cO\Paren{\log \Delta}$} \Comment{$\cO\Paren{\log \Delta}$ phases, or until $G$ is empty}
        \State Let prefix size $t_i = \cO\Paren{\frac{n \log n}{\Delta / 2^i}}$ and prefix offset $o_i = \sum_{z=0}^{i-1} t_z$.
        \State $G_i \leftarrow$ Prefix graph induced by vertices $\pi(o_i + 1), \ldots, \pi(o_i + t_i)$ with max.\ degree $\Delta'$.
        \State Process $G_i$ using \cref{alg:prefix-slower} or \cref{alg:prefix-faster}. \Comment{By Chernoff bounds, $\Delta' \in \cO(\log n)$}
    \EndFor
    \State Process any remaining vertices in $G$ using additional $\cO(\log \log n)$ MPC rounds.
\end{algorithmic}
\end{algorithm}

\begin{algorithm}[htbp]
\caption{Greedy MIS on $n$-vertex graph in $\cO\Paren{\log^2 \Delta \cdot \log \log n}$ MPC rounds in \cref{model:sublinear}}
\label{alg:prefix-slower}
\begin{algorithmic}[1]
    \State \textbf{Input}: Vertex ordering $\pi$, graph $G$ on $n$ vertices with maximum degree $\Delta$
    \For{$i = 0, 1, 2, \ldots, \ceil{\log_2 \Delta}$} \Comment{$\cO\Paren{\log \Delta}$ phases, or until $G$ is empty}
        \State Let chunk size $c_i = \frac{2^i}{100 \Delta} \cdot n$. \Comment{Chunk size doubles per phase}
        \For{$j = 1, 2, \ldots, 2000 \log \Delta$} \Comment{$\cO\Paren{\log \Delta}$ iterations, or until $G$ is empty}
            \State Let offset $o_{i,j} = c_i \cdot (j-1) + \sum_{z=0}^{i-1} c_z \cdot 2000 \log \Delta$.
            \State Let chunk graph $G_{i,j}$ be the graph induced by vertices $\pi\Paren{o_{i,j}}, ... , \pi\Paren{o_{i,j} + c_i}$.
            \State Process chunk graph $G_{i,j}$.
        \EndFor
    \EndFor
\end{algorithmic}
\end{algorithm}

\begin{remark}[Discussion about \cref{alg:prefix-slower}]
Our algorithm is inspired by the idea of graph shattering introduced by BEPS~\cite{barenboim2016locality}.
We break up the simulation of greedy MIS on $G_{\text{prefix}}$ into $\cO(\log \Delta)$ phases that process chunks of increasing size within the prefix graph.
By performing $\cO(\log \Delta)$ iterations within a phase, we can argue that any vertex in the remaining prefix graph has ``low degree'' with high probability in $\Delta$.
This allows us to prove that the connected components while processing every chunk of vertices is at most $\cO(\log n)$ and each vertex can learn (within the global memory limits) about the topology of its connected component in $\cO(\log \log n)$ rounds via graph exponentiation.
The constants 100 and 2000 are chosen for a cleaner analysis.
While the algorithm indexes more than $n$ vertices, we simply terminate after processing the last vertex in the permutation.\footnote{After all phases, we would have processed
$
\sum_{j=0}^{\ceil{\log_2 \Delta}} c_j \cdot 2000 \log \Delta
\geq \frac{2^{\ceil{\log_2 \Delta}}}{\Delta} \cdot n \cdot 2000 \log \Delta
\geq n
$
vertices.}
\end{remark}

\begin{algorithm}[htbp]
\caption{Greedy MIS on $n$-vertex graph in $\cO\Paren{\log \log n + \log \Delta}$ MPC rounds in \cref{model:sublinear-extra}}
\label{alg:prefix-faster}
\begin{algorithmic}[1]
    \State \textbf{Input}: Vertex ordering $\pi$, graph $G$ on $n$ vertices with maximum degree $\Delta$
    \State Assign a machine to each vertex.\Comment{In \cref{model:sublinear-extra}, we have $\geq n$ machines.}
    \State Graph exponentiate and gather $R$-hop neighborhood, where $R \in \cO\Paren{\frac{\log n}{\log \Delta}}$.
    \State Simulate greedy MIS (with respect to $\pi$) in $\cO\Paren{\log \Delta}$ MPC rounds.
\end{algorithmic}
\end{algorithm}

\begin{remark}[Discussion about \cref{alg:prefix-faster}]
We know from \cref{thm:dependency} that it suffices for each vertex know its $\cO(\log n)$-hop neighborhood in order to determine whether it is in the greedy MIS.
However, the $\cO(\log n)$-hop neighborhoods may not fit in a single machine.
So, we use \cref{sec:combination} to obtain a running time of $\cO\Paren{\log R + \frac{\log n}{R}} \subseteq \cO\Paren{\log \log n + \log \Delta}$.
\end{remark}

\begin{remark}[Comparison with the work of Blelloch, Fineman and Shun (BFS)~\cite{Blelloch2012}]
The algorithm of BFS~\cite{Blelloch2012} also considered prefixes of increasing size and they have a similar lemma as our \cref{lem:ordering-postfix}.
However, their work does not immediately imply ours.
The focus of BFS~\cite{Blelloch2012} was in the PRAM model in which the goal is to obtain an algorithm that is small work-depth --- they gave implementation of their algorithms that does a linear amount of work with polylogarithmic depth.
In this work, we are interested in studying the MPC model, in particular the sublinear memory regime.
Directly simulating their algorithm in MPC yields an algorithm that runs in $\cO(\log \Delta \cdot \log n)$ rounds.
Here, we crucially exploit graph exponentiation and round compression to speed up the greedy MIS simulation prefix graphs, enabling us to obtain algorithms that have an exponentially better dependency on $n$, i.e.\ that run in $\cO(\log \Delta \cdot \poly(\log \log n))$ rounds.
\end{remark}

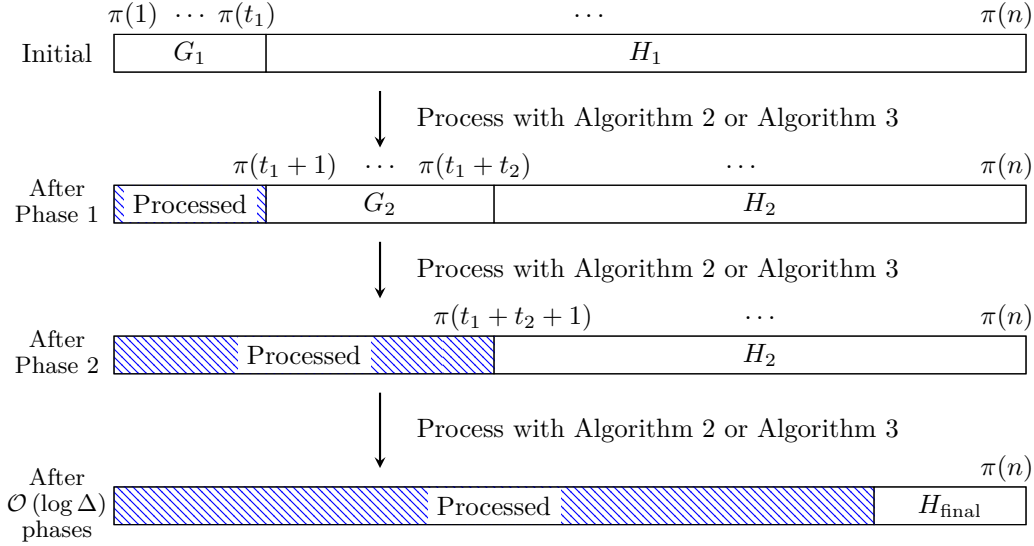
\begin{figure}[htbp]
\centering
\begin{tikzpicture}
\draw[] (0,0) rectangle (12,0.5);
\draw[] (0,-2) rectangle (12,-1.5);
\draw[] (0,-4) rectangle (12,-3.5);
\draw[] (0,-6) rectangle (12,-5.5);

\draw[] (0,0) rectangle (2,0.5) node[pos=0.5] {$G_1$};
\draw[] (2,0) rectangle (12,0.5) node[pos=0.5] {$H_1$};
\draw[] (2,-2) rectangle (5,-1.5) node[pos=0.5] {$G_2$};
\draw[] (5,-2) rectangle (12,-1.5) node[pos=0.5] {$H_2$};
\draw[] (5,-4) rectangle (12,-3.5) node[pos=0.5] {$H_2$};
\draw[] (10,-6) rectangle (12,-5.5) node[pos=0.5] {$H_{\text{final}}$};

\draw[pattern=north west lines, pattern color=blue] (0,-2) rectangle (2,-1.5) node[pos=0.5, fill=white] {Processed};
\draw[pattern=north west lines, pattern color=blue] (0,-4) rectangle (5,-3.5) node[pos=0.5, fill=white] {Processed};
\draw[pattern=north west lines, pattern color=blue] (0,-6) rectangle (10,-5.5) node[pos=0.5, fill=white] {Processed};

\node[] at (-0.75,0.25) {Initial};
\node[] at (-0.75,-1.75) {{\renewcommand{\arraystretch}{0.75}\begin{tabular}{c}After\\ Phase 1 \end{tabular}}};
\node[] at (-0.75,-3.75) {{\renewcommand{\arraystretch}{0.75}\begin{tabular}{c}After\\ Phase 2 \end{tabular}}};
\node[] at (-0.75,-5.75) {{\renewcommand{\arraystretch}{0.75}\begin{tabular}{c}After\\ $\cO\Paren{\log \Delta}$\\ phases\end{tabular}}};

\draw[thick, -stealth] (3.5, -0.25) -- node[right=10pt] {Process with \cref{alg:prefix-slower} or \cref{alg:prefix-faster}} (3.5, -1);
\draw[thick, -stealth] (3.5, -2.25) -- node[right=10pt] {Process with \cref{alg:prefix-slower} or \cref{alg:prefix-faster}} (3.5, -3);
\draw[thick, -stealth] (3.5, -4.25) -- node[right=10pt] {Process with \cref{alg:prefix-slower} or \cref{alg:prefix-faster}} (3.5, -5.25);

\node[] at (0.25,0.75) {$\pi(1)$};
\node[] at (1,0.75) {$\ldots$};
\node[] at (1.75,0.75) {$\pi(t_1)$};
\node[] at (6.25,0.75) {$\ldots$};
\node[] at (11.75,0.75) {$\pi(n)$};
\node[] at (2.25,-1.25) {$\pi(t_1 + 1)$};
\node[] at (3.5,-1.25) {$\ldots$};
\node[] at (4.75,-1.25) {$\pi(t_1 + t_2)$};
\node[] at (8.25,-1.25) {$\ldots$};
\node[] at (11.75,-1.25) {$\pi(n)$};
\node[] at (5.25,-3.25) {$\pi(t_1 + t_2 + 1)$};
\node[] at (8.5,-3.25) {$\ldots$};
\node[] at (11.75,-3.25) {$\pi(n)$};
\node[] at (11.75,-5.25) {$\pi(n)$};
\end{tikzpicture}
\caption{
Illustration of \cref{alg:greedy} given an initial graph $G$ on $n$ vertices with maximum degree $\Delta$.
Let $i \in \{1, \ldots, \cO(\log \Delta)\}$ and define $t_i = \cO\Paren{\frac{n \log n}{\Delta / 2^i}}$.
For each $i$, with high probability, the induced subgraph $G_i$ has maximum degree $\poly(\log n)$.
To process $G_i$, apply \cref{alg:prefix-slower} in $\cO(\log^3 \log n)$ MPC rounds, or \cref{alg:prefix-faster} in $\cO(\log \log n)$ MPC rounds while using extra global memory.
By our choice of $t_i$, \cref{lem:ordering-postfix} tells us that remaining subgraph $H_i$ has maximum degree $\Delta/2^i$.
We repeat this argument until the final subgraph $H_{\text{final}}$ involving $\poly(\log n)$ vertices, which can be processed in another call to \cref{alg:prefix-slower} or \cref{alg:prefix-faster}.
}
\label{fig:chunking}
\end{figure}

\subsection{Correlation clustering on bounded arboricity graphs}

Our algorithmic results for correlation clustering derive from the following key structural lemma that is proven by arguing that a local improvement to the clustering cost is possible if there exists large clusters.

\begin{lemma}[Structural lemma for correlation clustering (Informal)]
\label{lem:structural-informal}
There exists an optimum correlation clustering where all clusters have size at most $4 \arb - 2$.
\end{lemma}

This structural lemma allows us to perform cost-charging arguments against \emph{some} optimum clustering with bounded cluster sizes.
In particular, if a vertex has degree much larger than $\arb$, then many of its incident edges incur disagreements.
This insight yields the following algorithmic implication: we can effectively ignore high-degree vertices.

\begin{theorem}[Algorithmic implication (Informal)]
\label{thm:algo-implication}
Let $G$ be a graph where $E^+$ induces a $\arb$-arboric graph.
Form singleton clusters with vertices with degrees $\cO\Paren{\arb/\eps}$.
Run an $\alpha$-approximate algorithm $\cA$ on the remaining subgraph.
Then, the union of clusters is a $\max\{1 + \eps, \alpha\}$-approximation.
The runtime and approximation guarantees of the overall algorithm follows from the guarantees of $\cA$ (e.g.\ in expectation / worst case, det.\ / rand.).
\end{theorem}

Observe that \PIVOT essentially simulates a randomized greedy MIS with respect to a uniform-at-random ordering of vertices.
By setting $\eps = 2$ in \cref{thm:algo-implication} and $\Delta = \cO(\arb)$ in \cref{thm:randomized-greedy-MIS-informal}, we immediately obtain a 3-approximation (in expectation) algorithm for correlation clustering in $\cO\Paren{\log \arb \cdot \poly(\log \log n)}$ MPC rounds.
Note that we always have $\arb \leq \Delta \leq n$, and that $\arb$ can be significantly smaller than $\Delta$ and $n$ in general.
Many sparse graphs have $\arb \in \cO(1)$ while having unbounded maximum degrees, including planar graphs and bounded treewidth graphs.
As such, for several classes of graphs, our result improves over directly simulating \PIVOT in $\cO(\log n)$ rounds.

\begin{corollary}[General algorithm (Informal)]
\label{cor:general-informal}
Let $G$ be a complete signed graph such that $E^+$ induces a $\arb$-arboric graph.
There exists an algorithm that, with high probability, produces a 3-approximation (in expectation) for correlation clustering of $G$ in $\cO\Paren{\log \arb \cdot \log^3 \log n}$ MPC rounds in \cref{model:sublinear}, or $\cO\Paren{\log \arb \cdot \log \log n}$ MPC rounds in \cref{model:sublinear-extra}.
\end{corollary}

\begin{remark}[On converting ``in expectation'' to ``with high probability'']
Note that one can run $\cO(\log n)$ copies of \cref{cor:general-informal} in parallel and output the best clustering.
Applying this standard trick converts the ``in expectation'' guarantee to a ``with high probability'' guarantee with only a logarithmic factor increase in memory consumption.
\end{remark}

For forests with $\arb = 1$, \cref{lem:structural-informal} states that the optimum correlation clustering cost corresponds to the number of edges minus the size of the maximum matching.
Instead of computing a maximum matching, \cref{lem:apx-apx-informal} tells us that using an approximate matching suffices to obtain an $\alpha$-approximation (not necessarily maximal) to the correlation clustering problem.
Note that maximal matchings are 2-approximations and they always apply.

\begin{lemma}[Approximation via approximate matchings (Informal)]
\label{lem:apx-apx-informal}
Let $G$ be a complete signed graph such that $E^+$ induces a forest.
Suppose that the maximum matching size on $E^+$ is $\Abs{M^*}$.
If $M$ is a matching on $E^+$ such that $\alpha \cdot \Abs{M} \geq \Abs{M^*}$, for some $1 \leq \alpha \leq 2$, then clustering using $M$ yields an $\alpha$-approximation to the optimum correlation clustering of $G$.
\end{lemma}

Thus, it suffices to apply known maximum/approximate matching algorithms in sublinear memory regime of MPC to obtain correlation clustering algorithms in the special case of $\arb=1$.
More specifically, we consider the following results.
\begin{itemize}
    \item Using dynamic programming, BBDHM~\cite{bateni18trees} compute a maximum matching (on trees) in $\cO(\log n)$ MPC rounds.
    \item In the LOCAL model, EMR~\cite{GuyMatching2015} deterministically solve $(1+\eps)$-approx.\ matching in $\cO\Paren{\Delta^{\cO\Paren{\frac{1}{\eps}}} + \frac{1}{\eps^2} \cdot \log^* n}$ rounds.
    \item In the CONGEST model, BCGS~\cite{Yehuda2017} give an $\cO \Paren{2^{\cO(1/\eps)} \cdot \frac{\log \Delta}{\log \log \Delta}}$ round randomized algorithm for $(1+\eps)$-approx.\ matching.
\end{itemize}
These approximation results are heavily based on the Hopcroft-Karp framework~\cite{hopcroftKarp}, where independent sets of augmenting paths are iteratively flipped.
Since $\arb = 1$ and $\eps$ is a constant, we have a subgraph of constant maximum degree by ignoring vertices with degrees $\cO(\arb/\eps)$.
On this constant degree graph, each vertex only needs polylogarithmic memory when we perform graph exponentiation, satisfying the memory constraints of \cref{model:sublinear}.
Applying these matching algorithms together with \cref{thm:algo-implication} and \cref{lem:apx-apx-informal} yields the following result.

\begin{corollary}[Forest algorithm (Informal)]
\label{cor:forest-informal}
Let $G$ be a complete signed graph such that $E^+$ induces a forest and $0 < \eps \leq 1$ be a constant.
Hiding factors in $1/\eps$ using $\cO_{\eps}(\cdot)$, there exists:
\begin{enumerate}
    \item An optimum randomized algorithm that runs in $\cO(\log n)$ MPC rounds.
    \item A $(1+\eps)$-approximation (worst case) det.\ algo.\ that runs in $\cO_{\eps}\Paren{\log \log^* n}$ MPC rounds.
    \item A $(1+\eps)$-approximation (worst case) randomized algo.\ that runs in $\cO_{\eps}\Paren{1}$ MPC rounds.
\end{enumerate}
\end{corollary}

Finally, we give a simple $\cO\Paren{\arb^2}$-approximate (worst-case) algorithm in $\cO(1)$ MPC rounds.

\begin{corollary}[Simple algorithm (Informal)]
\label{cor:simple-informal}
Let $G$ be a complete signed graph such that $E^+$ induces a $\arb$-arboric graph.
Then, there exists an $\cO\Paren{\arb^2}$-approximation (worst case) deterministic algorithm that runs in $\cO(1)$ MPC rounds.
\end{corollary}

The simple algorithm is as follows: connected components which are cliques form clusters, and all other vertices form individual singleton clusters.
This can be implemented in $\cO(1)$ MPC rounds using broadcast trees.
We now give an informal argument when the input graph is a single connected component but not a clique.
By \cref{lem:structural-informal}, there will be $\geq n / \arb$ clusters and so the optimal number of disagreements is $\geq n / \arb$.
Meanwhile, the singleton clusters incurs errors on all positive edges, i.e.\ $\leq \arb \cdot n$ since $E^+$ induces a $\arb$-arboric graph.
Thus, the worst possible approximation ratio is $\approx \arb^2$.

\section{Randomized greedy MIS on bounded degree graphs}
\label{sec:randomized-greedy-mis}

In this section, we explain how to efficiently compute a randomized greedy MIS in the sublinear memory regime of the MPC model.
We will first individually analyze \cref{alg:prefix-slower} and \cref{alg:prefix-faster} and then show how to use them as black-box subroutines in \cref{alg:greedy}.
Both \cref{alg:prefix-slower} and \cref{alg:prefix-faster} rely on the result of Fischer and Noever~\cite{fischer2018tight} that it suffices for each vertex to know the $\pi$ ordering of its $\cO(\log n)$-hop neighborhood.
We defer the proofs of \cref{lem:small-comp}, \cref{lem:no-need-extra-global}, \cref{lem:ordering-prefix} and \cref{lem:ordering-postfix} to \cref{sec:deferred-proofs-greedy-mis}.

\cref{alg:prefix-slower} is inspired by the graph shattering idea introduced by BEPS~\cite{barenboim2016locality}.
Our analysis follows a similar outline as the analysis of the maximal independent set of BEPS~\cite{barenboim2016locality} but is significantly simpler as our ``vertex sampling process'' in each step simply follows the uniform-at-random vertex permutation $\pi$: we do not explicitly handle high-degree vertices at the end, but we argue that connected components are still small even if $\pi$ chooses some of them.
The key crux of our analysis is to argue that, for appropriately defined step sizes, the connected components considered are of size $\cO(\log n)$.

\begin{restatable}{lemma}{smallcomp}
\label{lem:small-comp}
Consider \cref{alg:prefix-slower}.
With high probability in $n$, the connected components in any chunk graph $G_{i,j}$ have size $\cO(\log n)$.
\end{restatable}

This allows us to argue that all vertices involved can learn the full topology of their connected components via graph exponentiation in $\cO(\log \log n)$ MPC rounds in \cref{model:sublinear}.

\begin{restatable}{lemma}{noneedextraglobal}
\label{lem:no-need-extra-global}
Consider \cref{alg:prefix-slower} in \cref{model:sublinear}.
Fix an arbitrary chunk graph $G_{i,j}$.
If connected components in $G_{i,j}$ have size at most $\poly(\log n)$, then every vertex can learn the full topology of its connected component in $\cO(\log \log n)$ MPC rounds.
\end{restatable}

\begin{lemma}
\label{lem:slower-main-statement}
Consider \cref{alg:prefix-slower} in \cref{model:sublinear}.
Suppose $G = (V,E)$ has $n$ vertices with maximum degree $\Delta$.
Let $\pi$ be a uniform-at-random ordering of $V$.
Then, with high probability, one can simulate greedy MIS on $G$ (with respect to $\pi$) in $\cO\Paren{\log^2 \Delta \cdot \log \log n}$ MPC rounds.
\end{lemma}
\begin{proof}
There are $\cO(\log \Delta)$ phases, each having $\cO(\log \Delta)$ iterations, in \cref{alg:prefix-slower}.
By \cref{lem:small-comp} and \cref{lem:no-need-extra-global}, each iteration can be computed in $\cO(\log \log n)$ MPC rounds.
\end{proof}

By using at least $n$ machines, \cref{alg:prefix-faster} presents a simpler and faster algorithm for computing greedy MIS compared to \cref{alg:prefix-slower}.
It exploits computational features of the MPC model such as graph exponentiation and round compression to speed up computation.

\begin{restatable}{lemma}{orderingprefix}
\label{lem:ordering-prefix}
Consider \cref{alg:prefix-faster} in \cref{model:sublinear-extra}.
Suppose $G = (V,E)$ has $n$ vertices with maximum degree $\Delta$.
Let $\pi$ be a uniform-at-random ordering of $V$.
Then, with high probability, one can simulate greedy MIS on $G$ (with respect to $\pi$) in $\cO\Paren{\log \log n + \log \Delta}$ MPC rounds.
\end{restatable}

Recall that \cref{alg:greedy} uses \cref{alg:prefix-slower} or \cref{alg:prefix-faster} as subroutines to compute the greedy MIS on a subgraph induced by some prefix of $\pi$'s ordering in each phase.
We first prove \cref{lem:ordering-postfix} which bounds the maximum degree of the remaining subgraph after processing $t \leq n$ vertices.
By our choice of prefix sizes, we see that the maximum degree is halved with high probability in each phase and thus $\cO(\log \Delta)$ phases suffice.

\begin{restatable}{lemma}{orderingpostfix}
\label{lem:ordering-postfix}
Let $G$ be a graph on $n$ vertices and $\pi: [n] \rightarrow V$ be a uniform-at-random ordering of vertices.
For $t \in [n]$, consider the subgraph $H_t$ obtained after processing vertices $\{\pi(1), \ldots, \pi(t)\}$ via greedy MIS (with respect to $\pi$).
Then, with high probability, the maximum degree in $H_t$ is at most $\cO\Paren{\frac{n \log n}{t}}$.
\end{restatable}

\begin{remark}
Similar statements to \cref{lem:ordering-prefix} and \cref{lem:ordering-postfix} were previously known.\footnote{E.g.\ see GGKMR~\cite[Section 3]{ghaffari2018improved}, ACGMW~\cite[Lemma 27]{Kook2015}, and BFS~\cite[Lemma 3.1]{Blelloch2012}.}
\end{remark}

\begin{theorem}
\label{thm:greedy-MIS-MPC}
Let $G$ be a graph with $n$ vertices of maximum degree $\Delta$ and $\pi: [n] \rightarrow V$ be a uniform-at-random ordering of vertices.
Then, with high probability, one can compute greedy MIS (with respect to $\pi$) in $\cO\Paren{\log \Delta \cdot \log^3 \log n}$ MPC rounds in \cref{model:sublinear}, or $\cO\Paren{\log \Delta \cdot \log \log n}$ MPC rounds in \cref{model:sublinear-extra}.
\end{theorem}
\begin{proof}
There are $\cO(\log \Delta)$ phases in \cref{alg:greedy}.
For $i \in \{1, \ldots, \cO(\log \Delta)\}$ and appropriate constant factors, we set $t_i = \cO\Paren{\frac{n \log n}{\Delta / 2^i}}$ and consider the induced prefix graph $G_i$.

There are two possible subroutines to process $G_i$: \cref{alg:prefix-slower} and \cref{alg:prefix-faster}.
By Chernoff bounds, the maximum degree in $G_i$ is $\cO(\log n)$ with high probability in $n$.
So, \cref{lem:slower-main-statement} and \cref{lem:ordering-prefix} tell us that \cref{alg:prefix-slower} and \cref{alg:prefix-faster} only need $\cO\Paren{\log^3 \log n}$ and $\cO\Paren{\log \log n}$ MPC rounds for a single invocation respectively.

By \cref{lem:ordering-postfix}, the maximum degree of the graph after processing $G_i$ is halved with high probability in $n$.
So, after $\cO(\log \Delta)$ phases, there are at most $\poly(\log n)$ vertices left in the graph.
Thus, the maximum degree is at most $\poly(\log n)$ and we apply the subroutine one last time.
Finally, since $\Delta \leq n$, we can apply union bound over these $\cO\Paren{\log \Delta}$ phases to upper bound the failure probability.
\end{proof}

\section{Structural properties for correlation clustering}
\label{sec:structural}

In this section, we prove our main result (\cref{thm:ignore-high-deg}) about correlation clustering by ignoring high-degree vertices.
To do so, we first show a structural result of optimum correlation clusterings (\cref{lem:arb-max-cluster-size}): there \emph{exists} an optimum clustering with bounded cluster sizes.
This structural lemma also implies that in the special case of forests (i.e.\ $\arb = 1$), a maximum matching on $E^+$ yields an optimum correlation clustering of $G$ (\cref{cor:maximum-matching}).

\begin{restatable}[Structural lemma for correlation clustering]{lemma}{arbmaxclustersize}
\label{lem:arb-max-cluster-size}
Let $G$ be a complete signed graph such that positive edges $E^+$ induce a $\arb$-arboric graph.
Then, there exists an optimum correlation clustering where all clusters have size at most $4 \arb - 2$.
\end{restatable}
\begin{proof}[Proof sketch]
The proof involves performing local updates by repeatedly removing vertices from large clusters while arguing that the number of disagreements does not increase (it may not strictly decrease but may stay the same).
See \cref{sec:deferred-proofs-structural} for details.
\end{proof}

\begin{restatable}[Algorithmic implication of \cref{lem:arb-max-cluster-size}]{theorem}{ignorehighdeg}
\label{thm:ignore-high-deg}
Let $G$ be a complete signed graph such that positive edges $E^+$ induce a $\arb$-arboric graph.
For $\eps > 0$, let
\[
H = \left\{v \in V : d(v) > \frac{8(1+\eps)}{\eps} \cdot \arb \right\} \subseteq V
\]
be the set of high-degree vertices, and $G' \subseteq G$ be the subgraph obtained by removing high-degree vertices in $H$.
Suppose $\cA$ is an $\alpha$-approximate correlation clustering algorithm and $cost(OPT(G))$ is the optimum correlation clustering cost.
Then,
\[
cost \left( \{\{v\} : v \in H\} \cup \cA(G') \right)
\leq \max \left\{ 1+\eps, \alpha \right\} \cdot cost(OPT(G))
\]
where $\{\{v\} : v \in H\} \cup \cA(G')$ is the clustering obtained by combining the singleton clusters of high-degree vertices with $\cA$'s clustering of $G'$.
See \cref{alg:general-alg} for a pseudocode.
Furthermore, if $\cA$ is $\alpha$-approximation only in expectation, then the above inequality holds only in expectation.
\end{restatable}
\begin{proof}[Proof sketch]
Fix an optimum clustering $OPT(G)$ of $G$ where each cluster has size at most $4 \arb - 2$.
Such a clustering exists by \cref{lem:arb-max-cluster-size}.
One can then show that $cost(OPT(G)) \geq \frac{1}{1+\eps} \cdot \Abs{M^+} + (\text{disagreements in $U$})$, where $\Abs{M^+}$ is the number of positive edges adjacent to high-degree vertices and $U$ is the set of edges \emph{not} adjacent to any high-degree vertex.
The result follows by combining singleton clusters of high-degree vertices $H$ and the $\alpha$-approximate clustering on low-degree vertices $U$ using $\cA$.
See \cref{sec:deferred-proofs-structural} for details.
\end{proof}

\begin{algorithm}[htbp]
\caption{Correlation clustering for $G$ such that $E^+$ induces a $\arb$-arboric graph}
\label{alg:general-alg}
\begin{algorithmic}[1]
    \State \textbf{Input}: Graph $G$, $\eps > 0$, $\alpha$-approximate algorithm $\cA$
    \State Let $H = \left\{v \in V : d(v) > \frac{8(1+\eps)}{\eps} \cdot \arb \right\} \subseteq V$ be the set of high-degree vertices.
    \State Let $G' \subseteq G$ be a bounded degree subgraph obtained by removing high-degree vertices $H$.
    \State Let $\cA(G')$ be the clustering obtained by running $\cA$ on the subgraph $G'$.
    \State \textbf{Return} Clustering $\{\{v\} : v \in H\} \cup \cA(G')$.
\end{algorithmic}
\end{algorithm}

\begin{restatable}[Maximum matchings yield optimum correlation clustering in forests]{corollary}{maximummatching}
\label{cor:maximum-matching}
Let $G$ be a complete signed graph such that positive edges $E^+$ induce a forest (i.e.\ $\arb = 1$).
Then, clustering using a \emph{maximum} matching on $E^+$ yields an optimum cost correlation clustering.
\end{restatable}
\begin{proof}
See \cref{sec:deferred-proofs-structural}.
\end{proof}
\section{Minimizing disagreements in bounded arboricity graphs and forests}
\label{sec:application}

We now describe how to use our main result (\cref{thm:ignore-high-deg}) to obtain efficient correlation clustering algorithms in the sublinear memory regime of the MPC model.
\cref{thm:ignore-high-deg} implies that we can focus on solving correlation clustering on graphs with maximum degree $\cO(\arb)$.

For general $\arb$-arboric graphs, we simulate \PIVOT by invoking \cref{thm:greedy-MIS-MPC} to obtain \cref{cor:general-alg}.
For forests, \cref{cor:maximum-matching} states that a maximum matching on $E^+$ yields an optimal correlation clustering.
Then, \cref{lem:forest-apx-bound} tells us that if one computes an approximate matching (not necessarily maximal) instead of a maximum matching, we still get a reasonable cost approximation to the optimum correlation clustering.
By invoking existing matching algorithms, we show how to obtain three different correlation clustering algorithms (with different guarantees) in \cref{cor:forest}.
Finally, \cref{cor:constant} gives a deterministic constant round algorithm that yields an $\cO(\arb^2)$ approximation.

\begin{corollary}
\label{cor:general-alg}
Let $G$ be a complete signed graph such that positive edges $E^+$ induce a $\arb$-arboric graph.
With high probability, there exists an algorithm that produces a 3-approximation (in expectation) for correlation clustering of $G$ in $\cO\Paren{\log \arb \cdot \log^3 \log n}$ MPC rounds in \cref{model:sublinear}, or $\cO\Paren{\log \arb \cdot \log \log n}$ MPC rounds in \cref{model:sublinear-extra}.
\end{corollary}
\begin{proof}
Run \cref{alg:general-alg} with $\eps = 2$ and \PIVOT as $\cA$.
The approximation guarantee is due to the fact that \PIVOT gives a 3-approximation in expectation.
Since $\eps = 2$, the maximum degree in $G'$ is $12 \arb$.
Set $\Delta = 12 \arb$ in \cref{thm:greedy-MIS-MPC}.
\end{proof}

\begin{lemma}
\label{lem:forest-apx-bound}
Let $G$ be a complete signed graph such that positive edges $E^+$ induce a forest.
Suppose $\Abs{M^*}$ is the size of a maximum matching on $E^+$ and $M$ is an approximate matching on $E^+$ where $\alpha \cdot \Abs{M} \geq \Abs{M^*}$ for some $1 \leq \alpha \leq 2$.
Then, clustering using $M$ yields an $\alpha$-approximation to the optimum correlation clustering of $G$.
\end{lemma}
\begin{proof}
Clustering based on any matching (i.e.\ forming clusters of size two for each matched pair of vertices and singleton clusters for unmatched vertices) incurs $n-1-\Abs{M}$ disagreements.
By \cref{cor:maximum-matching}, clustering with respect to a maximum matching yields a correlation clustering of optimum cost.
If $\Abs{M^*} = \Abs{M}$, then the approximation ratio is trivially 1.
Henceforth, let $\Abs{M} \leq \Abs{M^*} - 1$.
Observe that
$
\frac{n - 1 - \frac{1}{\alpha} \cdot \Abs{M^*}}{n - 1 - \Abs{M^*}} \leq \alpha
\iff \Abs{M^*} \cdot \Paren{1 + \frac{1}{\alpha}} \leq n - 1
$.
Indeed,
\begin{align*}
\Abs{M^*} \cdot \Paren{1 + \frac{1}{\alpha}}
& \leq \Abs{M^*} + \Abs{M} && \text{since $\alpha \cdot \Abs{M} \geq \Abs{M^*}$}\\
& \leq 2 \cdot \Abs{M^*} - 1 && \text{since $\Abs{M} \leq \Abs{M^*} - 1$}\\
& \leq n-1 && \text{since $\Abs{M^*} \leq \frac{n}{2}$ for any maximum matching}
\end{align*}
Thus, the approximation factor for using $M$ is
$
\frac{n - 1 - \Abs{M}}{n - 1 - \Abs{M^*}}
\leq \frac{n - 1 - \frac{1}{\alpha} \cdot \Abs{M^*}}{n - 1 - \Abs{M^*}}
\leq \alpha
$.
\end{proof}

\begin{remark}
The approximation ratio of \cref{lem:forest-apx-bound} tends to 1 as $\Abs{M}$ tends to $\Abs{M^*}$.
The worst ratio possible is 2 and this approximation ratio is tight:
consider a path of 4 vertices and 3 edges with $\Abs{M^*} = 2$ and maximal matching $\Abs{M} = 1$.
\end{remark}

\begin{restatable}{corollary}{forest}
\label{cor:forest}
Consider \cref{model:sublinear}.
Let $G$ be a complete signed graph such that positive edges $E^+$ induce a forest.
Let $0 < \eps \leq 1$ be a constant.
Then, there exists the following algorithms for correlation clustering:
\begin{enumerate}
    \item An optimum randomized algorithm that runs in $\widetilde{\cO}(\log n)$ MPC rounds.
    \item A $(1+\eps)$-approx.\ (worst case) deterministic algo.\ that runs in $\cO\Paren{\frac{1}{\eps} \cdot \Paren{\log \frac{1}{\eps} + \log \log^* n}}$ MPC rounds.
    \item A $(1+\eps)$-approx.\ (worst case) randomized algo.\ that runs in $\cO\Paren{\log \log \frac{1}{\eps}}$ MPC rounds.
\end{enumerate}
\end{restatable}
\begin{proof}[Proof sketch]
For (i), use the algorithm of BBDHM~\cite{bateni18trees}.
For (ii) and (iii), apply \cref{thm:ignore-high-deg} with $\arb = 1$, $\alpha = 1/(1+\eps)$, then use the deterministic algorithm of Even, Medina and Ron \cite{GuyMatching2015} and randomized algorithm of BCGS~\cite{Yehuda2017} respectively.
See \cref{sec:deferred-proofs-application} for details.
\end{proof}

\begin{restatable}{corollary}{constant}
\label{cor:constant}
Consider \cref{model:sublinear}.
Let $G$ be a complete signed graph such that positive edges $E^+$ induce a $\arb$-arboric graph.
Then, there exists a deterministic algorithm that produces an $\cO(\arb^2)$-approximation (worst case) for correlation clustering of $G$ in $\cO(1)$ MPC rounds.
\end{restatable}
\begin{proof}[Proof sketch]
Consider the following deterministic algorithm:
Each connected component (with respect to $E^+$) that is a clique forms a single cluster, then all remaining vertices form singleton clusters.
This can be implemented in $\cO(1)$ MPC rounds using broadcast trees.
For the approximation ratio, fix an optimum clustering $OPT(G)$ of $G$ such that each cluster has size at most $4 \arb - 2$.
Such a clustering exists by \cref{lem:arb-max-cluster-size}.
One can then show that our algorithm incurs a ratio of at most $\cO(\arb^2)$ for any arbitrary connected component $H$, with respect to $OPT(G)$.
See \cref{sec:deferred-proofs-application} for details.
\end{proof}

\begin{remark}
The approximation analysis in \cref{cor:constant} is tight (up to constant factors): consider the barbell graph where two cliques $K_{\arb}$ (cliques on $\arb$ vertices) are joined by a single edge.
The optimum clustering forms a cluster on each $K_{\arb}$ and incurs one external disagreement.
Meanwhile, forming singleton clusters incurs $\approx \arb^2$ positive disagreements.
\end{remark}

\section{Conclusions and open questions}
\label{sec:conclusion}

In this work, we present a structural result on correlation clustering of complete signed graphs such that the positive edges induce a bounded arboricity graph.
Combining this with known algorithms, we obtain efficient algorithms in the sublinear memory regime of the MPC model.
We also showed how to compute a \emph{randomized greedy MIS} in $\cO\Paren{\log \Delta \cdot \log \log n}$ MPC rounds.
As intriguing directions for future work, we pose the following questions:

\begin{question}
For graphs with maximum degree $\Delta \in \poly(\log n)$, can one compute greedy MIS in $\cO\Paren{\log \log n}$ MPC rounds in the sublinear memory regime of the MPC model?
\end{question}
For graphs with maximum degree $\Delta \in \poly(\log n)$, \cref{alg:prefix-slower} runs in $\cO\Paren{\log^3 \log n}$ MPC rounds and \cref{alg:prefix-faster} runs in $\cO\Paren{\log \log n}$ MPC rounds assuming access to at least $n$ machines.
Is it possible to achieve running time of $\cO\Paren{\log \log n}$ MPC rounds without additional global memory assumptions?

\begin{question}
Can a randomized greedy MIS be computed in $\cO\Paren{\log \Delta + \log \log n}$ or $\cO\Paren{\sqrt{\log \Delta} + \log \log n}$ MPC rounds?
\end{question}

This would imply that a 3-approximate (in expectation) correlation clustering algorithm in the same number of MPC rounds.
We posit that a better running time than $\cO\Paren{\log \Delta \cdot \log \log n}$ should be possible.
The informal intuition is as follows: Fischer and Noever's result \cite{fischer2018tight} tells us that most vertices do not have long dependency chains in every phase, so ``pipelining arguments'' might work.

\begin{question}
Is there an efficient \emph{distributed} algorithm to minimize disagreements with an approximation guarantee strictly better than 3 (in expectation), or worst-case guarantees for general graphs?
\end{question}

For minimizing disagreements in complete signed graphs, known algorithms (see \cref{sec:related}) with approximation guarantees strictly less than 3 (in expectation) are based on probabilistic rounding of LPs.
\emph{Can one implement such LPs efficiently in a distributed setting, or design an algorithm that is amenable to a distributed implementation with provable guarantees strictly better than 3?}
In this work, we gave algorithms with worst-case approximation guarantees when the graph induced by positive edges is a forest.
\emph{Can one design algorithms that give worst-case guarantees for general graphs?}

\bibliography{refs}

\appendix

\section{Deferred proofs}
\label{sec:deferred-proofs}

In this section, we provide the deferred proofs in the main text.
For convenience, we will restate them before giving the proofs.
We also prove any necessary intermediate results here.

\subsection{Proofs for \texorpdfstring{\cref{sec:randomized-greedy-mis}}{Section 3}}
\label{sec:deferred-proofs-greedy-mis}

\cref{lem:counting-connected} and \cref{lem:high-deg-prob} are intermediate results needed for \cref{lem:small-comp}.
A special case of \cref{lem:counting-connected} (with $s=0$) was first considered in BEPS~\cite{barenboim2016locality}.
Here, we provide more details while proving a slightly tighter bound.

\begin{lemma}[Counting connected neighborhoods]
\label{lem:counting-connected}
Consider a graph $G = (V,E)$ on $\Abs{V} = n$ vertices with maximum degree $\Delta \geq 1$.
Let $L \subseteq V$ be a subset of low-degree vertices with maximum degree $1 \leq \delta \leq \Delta$.
Then, the number of connected components involving $s$ vertices from $L$ and $t$ vertices from $V \setminus L$ is at most $n \cdot 4^{s+t-2} \cdot \delta^{s} \cdot \Delta^{t}$.
\end{lemma}
\begin{proof}
Since any connected component of size $(s+t)$ must contain a tree on $(s+t)$ nodes, counting the number of unlabelled $(s+t)$-node trees rooted at every vertex gives an \emph{upper bound} on the total number of $(s+t)$-node connected components in the graph.\footnote{A connected component may contain multiple trees and this upper bound is generally an over-estimation.}

One can count unlabelled trees using Euler tours.
Given a root in an $(s+t)$-node tree, there is a 1-to-1 correspondence between an Euler tour and a binary string $S$ of $2(s+t-2)$ bits.
See \cref{fig:tree-counting} for an illustration.

Given a rooted unlabelled tree $T$, we now upper bound the number of ways to assign labels to the unlabelled $(s+t-1)$ vertices by upper bounding the number of possible ways to embed $T$ into the graph $G$.
Consider labelling $T$ with an arbitrary ordering of $s$ vertices from $L$ and $t$ vertices from $V \setminus L$.
There are at most $\delta$ choices when branching off a vertex from $L$, while there are at most $\Delta$ choices when branching off a vertex from $V \setminus L$.
Thus, there are at most $\delta^{s} \cdot \Delta^{t}$ ways\footnote{Again, this is an overestimation.} to embed an arbitrary tree $T$ in $G$.

Putting everything together, we see that there are $n$ ways to root a tree, at most $2^{2(s+t-2)} = 4^{s+t-2}$ possible Euler tours, and at most $\delta^{s} \cdot \Delta^{t}$ ways to embed each tour in $G$.
That is, the number of connected components involving $s$ vertices from $L$ and $t$ vertices from $V \setminus L$ is at most $n \cdot 4^{s+t-2} \cdot \delta^{s} \cdot \Delta^{t}$.
\end{proof}

\newcommand{\basegraph}{
    \node[draw, circle, minimum size=5pt] at (0,0.75) (G-a) {$a$};
    \node[draw, circle, minimum size=5pt] at (-0.75,0) (G-b) {$b$};
    \node[draw, circle, minimum size=5pt] at (0.75,0) (G-c) {$c$};
    \node[draw, circle, minimum size=5pt] at (-0.35,-0.9) (G-d) {$d$};
    \node[draw, circle, minimum size=5pt] at (0.35,-0.9) (G-e) {$e$};
    \draw[] (G-a) -- (G-b);
    \draw[] (G-a) -- (G-c);
    \draw[] (G-a) -- (G-e);
    \draw[] (G-b) -- (G-c);
    \draw[] (G-b) -- (G-d);
}

\begin{figure}[htbp]
    \centering
    \begin{tabular}{@{}cccccc@{}}
    \\
    \multicolumn{6}{c}{
    \begin{tikzpicture}[baseline=0]
        \basegraph
        \node[] at (-2,0) {Graph $G$};
        \node[] at (2,0) {\phantom{Graph $G$}}; 
    \end{tikzpicture}}
    \\
    \\
    \begin{tikzpicture}[baseline=0]
        \basegraph
        \node[draw, thick, red, circle, minimum size=5pt] at (0,0.75) {$a$};
        \node[draw, thick, blue, circle, minimum size=5pt] at (-0.75,0) {$b$};
        \node[draw, thick, blue, circle, minimum size=5pt] at (-0.35,-0.9) {$d$};
        \draw[thick, blue] (G-a) -- node[above, sloped, pos=0.2]{\scriptsize $\leftarrow$} node[below, sloped, pos=0.8]{\scriptsize $\rightarrow$} (G-b);
        \draw[thick, blue] (G-b) -- node[above, sloped, pos=0.8]{\scriptsize $\leftarrow$} node[below, sloped, pos=0.2]{\scriptsize $\rightarrow$} (G-d);
    \end{tikzpicture}
    &
    \begin{tikzpicture}[baseline=0]
        \basegraph
        \node[draw, thick, red, circle, minimum size=5pt] at (0,0.75) {$a$};
        \node[draw, thick, blue, circle, minimum size=5pt] at (-0.75,0) {$b$};
        \node[draw, thick, blue, circle, minimum size=5pt] at (0.75,0) {$c$};
        \draw[thick, blue] (G-a) -- node[above, sloped, pos=0.2]{\scriptsize $\leftarrow$} node[below, sloped, pos=0.8]{\scriptsize $\rightarrow$} (G-b);
        \draw[thick, blue] (G-a) -- node[below, sloped, pos=0.8]{\scriptsize $\leftarrow$} node[above, sloped, pos=0.2]{\scriptsize $\rightarrow$} (G-c);
    \end{tikzpicture}
    &
    \begin{tikzpicture}[baseline=0]
        \basegraph
        \node[draw, thick, red, circle, minimum size=5pt] at (0,0.75) {$a$};
        \node[draw, thick, blue, circle, minimum size=5pt] at (-0.75,0) {$b$};
        \node[draw, thick, blue, circle, minimum size=5pt] at (0.75,0) {$c$};
        \draw[thick, blue] (G-a) -- node[above, sloped, pos=0.2]{\scriptsize $\leftarrow$} node[below, sloped, pos=0.8]{\scriptsize $\rightarrow$} (G-b);
        \draw[thick, blue] (G-b) -- node[above, sloped, pos=0.8]{\scriptsize $\leftarrow$} node[below, sloped, pos=0.2]{\scriptsize $\rightarrow$} (G-c);
    \end{tikzpicture}
    &
    \begin{tikzpicture}[baseline=0]
        \basegraph
        \node[draw, thick, red, circle, minimum size=5pt] at (0,0.75) {$a$};
        \node[draw, thick, blue, circle, minimum size=5pt] at (-0.75,0) {$b$};
        \node[draw, thick, blue, circle, minimum size=5pt] at (0.75,0) {$c$};
        \draw[thick, blue] (G-a) -- node[below, sloped, pos=0.8]{\scriptsize $\leftarrow$} node[above, sloped, pos=0.2]{\scriptsize $\rightarrow$} (G-c);
        \draw[thick, blue] (G-b) -- node[below, sloped, pos=0.85]{\scriptsize $\leftarrow$} node[above, sloped, pos=0.2]{\scriptsize $\rightarrow$} (G-c);
    \end{tikzpicture}
    &
    \begin{tikzpicture}[baseline=0]
        \basegraph
        \node[draw, thick, red, circle, minimum size=5pt] at (0,0.75) {$a$};
        \node[draw, thick, blue, circle, minimum size=5pt] at (-0.75,0) {$b$};
        \node[draw, thick, blue, circle, minimum size=5pt] at (0.35,-0.9) {$e$};
        \draw[thick, blue] (G-a) -- node[above, sloped, pos=0.2]{\scriptsize $\leftarrow$} node[below, sloped, pos=0.8]{\scriptsize $\rightarrow$} (G-b);
        \draw[thick, blue] (G-a) -- node[above, sloped, pos=0.8]{\scriptsize $\leftarrow$} node[below, sloped, pos=0.2]{\scriptsize $\rightarrow$} (G-e);
    \end{tikzpicture}
    &
    \begin{tikzpicture}[baseline=0]
        \basegraph
        \node[draw, thick, red, circle, minimum size=5pt] at (0,0.75) {$a$};
        \node[draw, thick, blue, circle, minimum size=5pt] at (0.75,0) {$c$};
        \node[draw, thick, blue, circle, minimum size=5pt] at (0.35,-0.9) {$e$};
        \draw[thick, blue] (G-a) -- node[below, sloped, pos=0.8]{\scriptsize $\leftarrow$} node[above, sloped, pos=0.2]{\scriptsize $\rightarrow$} (G-c);
        \draw[thick, blue] (G-a) -- node[above, sloped, pos=0.8]{\scriptsize $\leftarrow$} node[below, sloped, pos=0.2]{\scriptsize $\rightarrow$} (G-e);
    \end{tikzpicture}
    \\
    abdba & abaca / acaba & abcba & acbca & abaea / aeaba & acaea / aeaca\\
    1100 & 1010 & 1100 & 1100 & 1010 & 1010
    \end{tabular}
    \caption{
    Consider graph $G$.
    Starting from root vertex $a$, we annotate all possible Euler tours involving $k=3$ vertices with a corresponding $2(k-2)$-bit string.
    Each of the $k-1$ edges corresponds to 2 bits: a ``1'' for ``moving down the tree'' and a ``0'' for ``moving up the tree''.
    As each Euler tour must return to the root, the first bit is always a ``1'', the last bit is always a ``0'', and there are an equal number of 1's and 0's.
    Furthermore, there are at least as many 1's as 0's in \emph{any} prefix of the binary string.
    Finally, observe that while $2(k-2)$ bits uniquely identifies a rooted Euler tour tree involving $k$ vertices, it could have multiple embeddings in $G$.
    }
    \label{fig:tree-counting}
\end{figure}
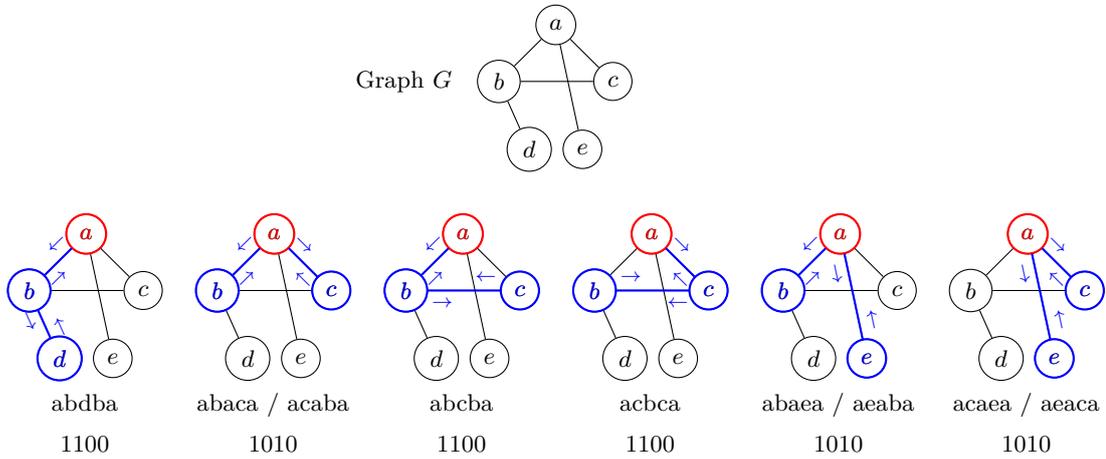

\begin{lemma}
\label{lem:high-deg-prob}
Let vertex $v$ be an arbitrary vertex in phase $i \in \{0, 1, \ldots, \ceil{\log_2 \Delta}\}$.
That is, $v \in \bigcup_{j=1}^{2000 \log \Delta} V \Paren{G_{i,j}}$.
Then,
\[
\Pr\Brac{\text{$v$ has $> \frac{\Delta}{2^{i-1}}$ neighbors in phases $i, i+1, \ldots, \ceil{\log_2 \Delta}$}}
\leq \Delta^{-10}
\]
\end{lemma}
\begin{proof}
When $i=0$, the statement trivially holds because all vertices have $\leq \Delta$ neighbors.
In phase $i > 1$, we have processed at least $\frac{2^i}{100 \Delta} \cdot n \cdot 2000 \log \Delta$ vertices.
If $v$ has $> \frac{\Delta}{2^{i-1}}$ neighbors in phases $i, i+1, \ldots, \ceil{\log_2 \Delta}$, then none of these neighbors belonged to the earlier phases.
Since $\pi$ is a uniform-at-random permutation, we see that
\begin{multline*}
\Pr\Brac{\text{$v$ has $> \frac{\Delta}{2^{i-1}}$ neighbors in phases $i, i+1, \ldots, \ceil{\log_2 \Delta}$}}\\
\leq \Paren{1 - \frac{2^i}{100 \Delta} \cdot 2000 \log \Delta}^{\frac{\Delta}{2^{i-1}}}
\leq \Delta^{-10}
\end{multline*}
\end{proof}

\smallcomp*
\begin{proof}[Proof of \cref{lem:small-comp}]
Partition the vertices in $G$ into $L$ and $V \setminus L$, where $L$ is the set of vertices in $G_{i,j}$ with degree $\leq \frac{\Delta}{2^{i-1}}$ and $V \Paren{G_{i,j}} \setminus L$ be the set of vertices in $G_{i,j}$ with degree $> \frac{\Delta}{2^{i-1}}$.
For arbitrary choices of $s$ and $t$, let us denote $C_{s,t}$ as a connected component involving $s$ vertices from $L$ and $t$ vertices from $V \setminus L$.
We will argue that the probability of a connected component $C_{s,t}$ existing in $G_{i,j}$ is very small when $s+t = 100 \log n \in \cO(\log n)$.
This implies our desired statement since any connected component of size at least $100 \log n$ must contain some $C_{s,t}$ with $s+t = 100 \log n$.

Fix an arbitrary connected component $C_{s,t}$.
Since $\pi$ is a uniform-at-random permutation, vertices are assigned to $G_{i,j}$ independently with probability $\frac{c_i}{n}$.
By \cref{lem:high-deg-prob}, the probability of a vertex in $G_{i,j}$ belonging in $T$ is at most $\Delta^{-10}$.
Thus, the probability of $C_{s,t}$ appearing in $G_{i,j}$ is at most $\Paren{\frac{c_i}{n}}^{s+t} \cdot \Delta^{-10 t}$.

By \cref{lem:counting-connected}, the number of connected components $C_{s,t}$ is at most $n \cdot 4^{s+t-2} \cdot \Paren{\frac{\Delta}{2^{i-1}}}^{s} \cdot \Delta^{t}$.
Recall that $i \leq \ceil{\log_2 \Delta}$.
So, by a union bound over all possible connected components $C_{s,t}$ with $s+t = 100 \log n$, we see that
\begin{align*}
&\; \Pr\Brac{\text{There exists $C_{s,t}$ with $s+t = 100 \log n$ in $G_{i,j}$}}\\
\leq &\; n \cdot 4^{s+t-2} \cdot \Paren{\frac{\Delta}{2^{i-1}}}^{s} \cdot
\Delta^{t} \cdot \Paren{\frac{c_i}{n}}^{s+t} \cdot \Delta^{-10 t}\\
= &\; n \cdot 4^{s+t-2} \cdot \Paren{\frac{\Delta}{2^{i-1}}}^{s} \cdot \Paren{\frac{2^i}{100 \Delta}}^{s+t} \cdot \Delta^{-9t}\\
\leq &\; n \cdot \Paren{\frac{4}{50}}^{s} \cdot \Delta^{-9t}\\
\leq &\; n^{-5}
\end{align*}
That is, with high probability in $n$, the connected components in $G_{i,j}$ have size $\cO(\log n)$.
\end{proof}

In the above proof of \cref{lem:small-comp}, we ignored the fact that some of the vertices in $G_{i,j}$ may have already been removed from the graph due to some neighbor in earlier phases entering the independent set.
However, this can only reduce the size of the connected components that we considered and is beneficial for the purposes of our analysis.

\noneedextraglobal*
\begin{proof}[Proof of \cref{lem:no-need-extra-global}]
Let $h: [n] \rightarrow [M]$ be a hash function chosen uniformly at random from a family $H$ of pairwise independent hash functions, where $M = n \cdot S$ is the number of machines and $S \in \widetilde{\cO}\Paren{n^{\delta}}$.
Let us assign each vertex $v$ to machine indexed by $h(v)$.
By Chernoff bounds, the number of vertices assigned to each machine is $\cO\Paren{S}$ with high probability.
By \cref{lem:small-comp}, the connected components have size $\cO(\log n)$.
To learn about the full topology of their connected component, we perform in $\cO(\log \log n)$ rounds of graph exponentiation.
Since each connected component only has size $\poly(\log n)$, each vertex requires at most $\poly(\log n)$ space whilst performing graph exponentiation and thus the total amount of memory used by \emph{any} collection of vertices on the same machine fits into the machine memory.
\end{proof}

\orderingprefix*
\begin{proof}[Proof of \cref{lem:ordering-prefix}]
By \cref{thm:dependency}, it suffices for any vertex to learn the ordering of vertices within its $\cO(\log n)$-hop neighborhood in $G_{\text{prefix}}$ to determine whether itself is in the MIS.
However, due to machine memory constraints, vertices may not be able to directly store their full $\cO(\log n)$-hop neighborhoods in a single machine.
Instead, vertices will gather their $\cO\Paren{\frac{\log n}{\log \Delta}}$-hop neighborhood via graph exponentiation, then simulate the greedy MIS algorithm in $\cO\Paren{\log \Delta}$ compressed rounds.
The total runtime of this procedure is $\cO\Paren{\log \log n + \log \Delta}$.

It remains to show that the $\cO\Paren{\frac{\log n}{\log \Delta}}$-hop neighborhood of any vertex fits in a single machine.
Let us be more precise about the constant factors involved.
Suppose that the longest dependency chain in greedy MIS has length $L \cdot \log n$ for some constant $L > 0$.
For some constant $C > 0$, let $R = C \cdot L \cdot \Paren{\frac{\log n}{\log \Delta}}$ denote the $R$-hop neighborhood that we want to collect into a single machine with memory $\widetilde{\cO}\Paren{n^\delta}$.
Note that the parameters $L$ and $\delta$ are given to us, and we only have control over the parameter $C$.
If we pick $C$ such that $C \cdot L < \delta < 1$, then
\[
R \cdot \log \Delta
= C \cdot L \cdot \Paren{\frac{\log n}{\log \Delta}} \cdot \log \Delta
= C \cdot L \log n
\in \cO\Paren{\delta \cdot \log n}
\iff
\Delta^R \in \cO\Paren{n^\delta}
\]
Thus, with appropriate constant factors, the $\cO\Paren{\frac{\log n}{\log \Delta}}$-hop neighborhood of any vertex fits in a single machine.
\end{proof}

\orderingpostfix*
\begin{proof}[Proof of \cref{lem:ordering-postfix}]
For the sake of clarity, we now prove the statement by setting the maximum degree bound in $H_t$ to $\frac{10 n \log n}{t}$.
The constant 10 is arbitrary and can be adjusted based on how this lemma is invoked.

Let $t' \in [t]$ be an arbitrary round and $v$ an arbitrary vertex.
Suppose $v$ has degree $d_{t'-1}(v)$ \emph{after} processing the first $(t'-1)$ vertices defined by $\pi(1), \ldots, \pi(t' - 1)$.
If $t' = 1$, then nothing has been processed yet and $d_{t'-1}(v) = d_{0}(v) = d(v)$, where $d(v)$ is the degree of $v$ in the input graph $G$.
If $d_{t'-1}(v) \leq \frac{10 n \log n}{t}$, then the vertex $v$ already satisfies the lemma since vertex degrees never increase while processing $\pi$.
Otherwise, $d_{t'-1}(v) > \frac{10 n \log n}{t}$.
We now proceed to upper bound the probability of vertex $v$ remaining in the subgraph after processing $\pi(t')$.

For vertex $v$ to remain, neither $v$ nor any of its neighbors must be chosen to be $\pi(t'-1)$.
Since $\pi$ is a uniform-at-random ordering of vertices, this happens with probability
$
1 - \frac{1 + d_{t'-1}(v)}{n - t' + 1}
\leq 1 - \frac{d_{t'-1}(v)}{n}
< 1 - \frac{10 n \log n}{tn}
= 1 - \frac{10 \log n}{t}
$.
Thus, the probability that vertex $v$ remains in $H_t$, while having $d_{t}(v) > \frac{10 n \log n}{t}$ after processing $\pi(t)$, is at most
$
\Paren{1 - \frac{10 \log n}{t}}^t
\leq \exp \Paren{- 10 \log n}
= n^{-10}
$.
The lemma follows by taking a union bound over all vertices.
\end{proof}

\subsection{Proofs for \texorpdfstring{\cref{sec:structural}}{Section 4}}
\label{sec:deferred-proofs-structural}

\arbmaxclustersize*
\begin{proof}[Proof of \cref{lem:arb-max-cluster-size}]
The proof involves performing local updates by repeatedly removing vertices from large clusters while arguing that the number of disagreements does not increase (it may not strictly decrease but may stay the same).

Consider an arbitrary clustering that has a cluster $C$ of size at least $\Abs{C} \geq 4 \arb - 1$.
We will show that there exists some vertex $v^* \in C$ such that $d^+_C(v^*) \leq 2 \arb - 1$.
Observe that removing $v^*$ to form its own singleton cluster creates $d^+_C(v^*)$ positive disagreements and removes $(\Abs{C}-1) - d^+_C(v^*)$ negative disagreements.
Since $d^+_C(v^*) \leq 2 \arb - 1 \leq \frac{\Abs{C} + 1}{2} - 1 = \frac{\Abs{C}-1}{2}$, we see that this local update will not increase the number of disagreements.
It remains to argue that $v^*$ exists.

Suppose, for a contradiction, that such a vertex $v^*$ does not exist in a cluster of size $\Abs{C} \geq 4 \arb - 1$.
Then, $d^+_C(v) \geq 2 \arb$ for \emph{all} vertices $v \in C$.
Summing over all vertices in $C$, we see that
\[
\Abs{E^+(C)}
= \frac{1}{2} \sum_{v \in C} d^+_C(v)
\geq \frac{1}{2} \cdot \Abs{C} \cdot 2 \arb
= \Abs{C} \cdot \arb
\geq \Abs{C} \cdot \frac{\Abs{E^+(C)}}{\Abs{C}-1}
> \Abs{E^+(C)}
\]
where the second last inequality follows from the definition of arboricity.
This is a contradiction, thus such a vertex $v^*$ exists.
Repeating this argument (i.e.\ removing all vertices like $v^*$ to form their own singleton clusters), we can transform any optimum clustering into one with clusters of size at most $4 \arb - 2$.
\end{proof}

\ignorehighdeg*
\begin{proof}[Proof of \cref{thm:ignore-high-deg}]
Denote edges incident to high-degree vertices as \emph{marked} ($M$), and \emph{unmarked} ($U$) otherwise.
We further split marked edges into positive ($M^+$) and negative marked edges ($M^-$).
In other words, we partition the edge set $E$ into $M^+ \cup M^- \cup U$.
Instead of the usual handshaking lemma\footnote{Handshaking lemma: $\sum_{v \in V} d(v) = 2 \Abs{E}$}, we have
\begin{equation}
\label{eq:marked-bound}
\Abs{M^+} \leq \sum_{v \in H} d^+(v) \leq 2 \cdot \Abs{M^+}
\end{equation}
because high-degree vertices may have low-degree neighbors and marked edges may be counted twice in the sum.
See \cref{fig:marked-inequality} for an illustration.

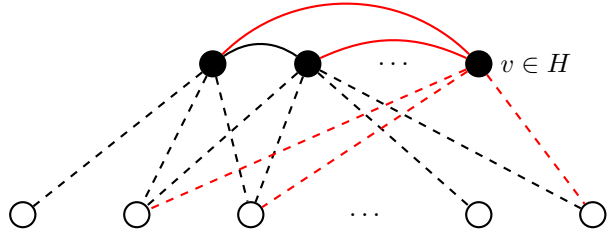
\begin{figure}[htbp]
\centering
\begin{tikzpicture}
\node[draw, circle, thick, fill=black] at (-0.5,0) (m1) {};
\node[draw, circle, thick, fill=black] at (0.75,0) (m2) {};
\node[] at (1.875,0) {$\ldots$};
\node[draw, circle, thick, fill=black] at (3,0) (m3) {};
\node[draw, circle, thick] at (-3,-2) (u1) {};
\node[draw, circle, thick] at (-1.5,-2) (u2) {};
\node[draw, circle, thick] at (0,-2) (u3) {};
\node[] at (1.5,-2) {$\ldots$};
\node[draw, circle, thick] at (3,-2) (u4) {};
\node[draw, circle, thick] at (4.5,-2) (u5) {};

\draw[thick, dashed] (m1) -- (u1);
\draw[thick, dashed] (m1) -- (u2);
\draw[thick, dashed] (m1) -- (u3);
\draw[thick, dashed] (m2) -- (u5);
\draw[thick, dashed] (m2) -- (u2);
\draw[thick, dashed] (m2) -- (u3);
\draw[thick, dashed] (m2) -- (u4);
\draw[thick, dashed, red] (m3) -- (u3);
\draw[thick, dashed, red] (m3) -- (u2);
\draw[thick, dashed, red] (m3) -- (u5);
\draw[thick, red] (m1) to[out=45,in=135] (m3);
\draw[thick] (m1) to[out=35,in=145] (m2);
\draw[thick, red] (m2) to[out=25,in=155] (m3);

\node[] at (3.75,0) {$v \in H$};
\end{tikzpicture}
\caption{
High-degree vertices $H$ are filled and only edges in $M^+$ are shown.
Edges contributing to $d^+(v)$ are highlighted in red.
In the summation of \cref{eq:marked-bound}, dashed edges are counted only once and solid edges are counted twice, hence $\Abs{M^+} \leq \sum_{v \in H} d^+(v) \leq 2 \Abs{M^+}$.
}
\label{fig:marked-inequality}
\end{figure}

Fix an optimum clustering $OPT(G)$ of $G$ where each cluster has size at most $4 \arb - 2$.
Such a clustering exists by \cref{lem:arb-max-cluster-size}.
Observe that
\begin{align*}
&\; cost(OPT(G))\\
= &\; (\text{disagreements in $M^+$}) + (\text{disagreements in $M^-$}) + (\text{disagreements in $U$})\\
\geq &\; (\text{disagreements in $M^+$}) + 0 + (\text{disagreements in $U$})
\end{align*}
We defer the proof for the following inequality and first use it to derive our result:
\begin{equation}
\label{eq:opt-lower-bound}
cost(OPT(G)) \geq \frac{1}{1+\eps} \cdot \Abs{M^+} + (\text{disagreements in $U$})
\end{equation}
Since ignoring high-degree vertices in the clustering from $OPT(G)$ yields a valid clustering for $G'$, we see that $cost(OPT(G')) \leq (\text{disagreements in $U$})$.
Thus,
\begin{align*}
&\; cost \left( \{\{v\} : v \in H\} \cup \cA(G') \right)\\
= &\; \Abs{M^+} + cost(\cA(G'))\\
\leq &\; \Abs{M^+} + \alpha \cdot cost(OPT(G')) && \text{since $\cA$ is $\alpha$-approximate}\\
\leq &\; \Abs{M^+} + \alpha \cdot (\text{disagreements in $U$}) && \text{since $cost(OPT(G')) \leq (\text{disagreements in $U$})$}\\
\leq &\; \max \left\{ 1+\eps, \alpha \right\} \cdot cost(OPT(G)) && \text{by \cref{eq:opt-lower-bound}}
\end{align*}
If $\cA$ is only $\alpha$-approximate in expectation, the same argument yields the same conclusion, but only in expectation.

To prove \cref{eq:opt-lower-bound}, it suffices to show that $(\text{disagreements in $M^+$}) \geq \frac{1}{1+\eps} \cdot \Abs{M^+}$.
Consider an arbitrary high-degree vertex $v \in H$.
By choice of $OPT(G)$ and definition of $H$, at most
\begin{equation}
\label{eq:bound-on-good-edges}
4 \arb - 2
\leq 4 \arb
\leq \frac{\eps}{2 (1 + \eps)} \cdot d(v)    
\end{equation}
of marked edges incident to $v$ belong in the same cluster and will not cause any external positive disagreements.
Let us call these edges \emph{good edges} since they do not incur any cost in $cost(OPT(G))$.
In total, across all high-degree vertices, there are at most
$
\sum_{v \in H} \frac{\eps}{2 (1 + \eps)} \cdot d(v)
\leq \frac{\eps}{2 (1 + \eps)} \cdot 2 \cdot \Abs{M^+}
= \frac{\eps}{1 + \eps} \cdot \Abs{M^+}
$
good edges due to \cref{eq:bound-on-good-edges} and \cref{eq:marked-bound}.
In other words,
$
(\text{disagreements in $M^+$})
\geq \Paren{1 - \frac{\eps}{1 + \eps}} \cdot \Abs{M^+}
= \frac{1}{1+\eps} \cdot \Abs{M^+}
$.
\end{proof}

\maximummatching*
\begin{proof}[Proof of \cref{cor:maximum-matching}]
By \cref{lem:arb-max-cluster-size}, with $\arb = 1$, we know that there exists an optimum correlation clustering where clusters are of size 1 or 2.
A size-2 cluster reduces one disagreement if the vertices are joined by a positive edge.
Hence, the total number of disagreements is minimized by maximizing the number of such size-2 clusters, which is the same as computing the \emph{maximum} matching on the set of positive edges $E^+$.
\end{proof}

\subsection{Proofs for \texorpdfstring{\cref{sec:application}}{Section 5}}
\label{sec:deferred-proofs-application}

\forest*
\begin{proof}[Proof of \cref{cor:forest}]
For the first algorithm, we use the algorithm of BBDHM~\cite{bateni18trees} to compute a maximum matching in $\widetilde{\cO}(\log n)$ MPC rounds, and cluster matched vertices together according to \cref{cor:maximum-matching}.

For the second algorithm, we apply \cref{thm:ignore-high-deg} with $\arb = 1$, $\alpha = 1/(1+\eps)$, and $\cA$ as the deterministic approximate matching algorithm of Even, Medina and Ron \cite{GuyMatching2015} on the subgraph with maximum degree $\Delta \in \cO\Paren{1/\eps}$.
Their algorithm runs in $R \in \cO\Paren{\Delta^{\cO\Paren{1/\eps}} + \frac{1}{\eps^2} \cdot \log^* n}$ LOCAL rounds.
This can be sped up to $\cO\Paren{\frac{1}{\eps} \cdot \Paren{\log \frac{1}{\eps} + \log \log^* n}}$ MPC rounds via graph exponentiation since $\Delta \in \cO\Paren{1/\eps}$ and each $R$-hop neighborhood is of logarithmic size.

For the third algorithm, we apply \cref{thm:ignore-high-deg} with $\arb = 1$, $\alpha = 1 + \eps$, and $\cA$ as the randomized approximate matching algorithm of BCGS~\cite{Yehuda2017} on the subgraph with maximum degree $\Delta \in \cO\Paren{1/\eps}$.
Their algorithm runs in $R \in \cO\Paren{\frac{\log \frac{1}{\eps}}{\log \log \frac{1}{\eps}}} \subseteq \cO\Paren{\log \frac{1}{\eps}}$ CONGEST rounds.
This can be sped up to $\cO\Paren{\log \log \frac{1}{\eps}}$ MPC rounds via graph exponentiation since $\Delta \in \cO\Paren{1/\eps}$ and each $R$-hop neighborhood is of logarithmic size.
\end{proof}

\begin{remark}
Let $M^*$ be some maximum matching and $M$ be some approximate matching.
The approximation ratio in \cite{GuyMatching2015} is stated as $\Abs{M} = (1-\eps') \cdot \Abs{M^*}$ while we write $(1 + \eps) \cdot \Abs{M} = \Abs{M^*}$.
This is only a constant factor difference: $\eps' \in \Theta(\eps)$.
\end{remark}

\constant*
\begin{proof}
Consider the following deterministic algorithm: 
Each connected component (with respect to $E^+$) that is a clique forms a single cluster, then all remaining vertices form singleton clusters.

\smallskip
\noindent\textbf{MPC implementation.}
Any clique in a $\arb$-arboric graph involves at most $2 \arb$ vertices.
Ignoring vertices with degrees larger than $(2 \arb - 1)$, the algorithm can be implemented in $\cO(1)$ MPC rounds using broadcast trees.

\smallskip
\noindent\textbf{Approximation analysis.}
Fix an optimum clustering $OPT(G)$ of $G$ such that each cluster has size at most $4 \arb - 2$.
Such a clustering exists by \cref{lem:arb-max-cluster-size}.
Note that clusters in \emph{any} optimum clustering are connected components (with respect to $E^+$), otherwise one can strictly improve the cost by splitting up such clusters.
By bounding the approximation ratio for an arbitrary connected component in the input graph $G$, we obtain a worst case approximation ratio.

Consider an arbitrary connected component $H$ on $n$ vertices and $m$ positive edges.
Since $H$ is $\arb$-arboric, $m \leq \arb \cdot n$.
If $H$ is a clique, then our algorithm incurs zero disagreements.
Otherwise, our algorithm forms singleton clusters and incurs $m \leq \arb \cdot n$ disagreements.
Since each cluster in $OPT(G)$ has size at most $4 \arb - 2$, there must be at least $\frac{n}{4 \arb - 2}$ clusters of $OPT(G)$ involving vertices in $H$.
Since $H$ is a connected component, this means that $OPT(G)$ incurs at least $\frac{n}{4 \arb - 2} - 1$ positive external disagreements with respect to $H$.
Thus, the worst possible approximation ratio is $\frac{\arb \cdot n}{\frac{n}{4 \arb - 2} - 1} \in \cO(\arb^2)$.
\end{proof}

\end{document}